\documentclass[preprint,aps,floatfix,nofootinbib,a4paper]{revtex4-1}

\usepackage{amsmath,amssymb,amsfonts,amsthm,latexsym,epsfig,mathrsfs,xcolor}
\usepackage[hidelinks]{hyperref}

\renewcommand\thesection{\arabic{section}}
\renewcommand\thesubsection{\arabic{subsection}}

\makeatletter
\def\p@subsection{\thesection.}
\def\p@subsubsection{\thesection.\thesubsection.}
\makeatother 

\numberwithin{equation}{section}

\theoremstyle{definition}
\newtheorem{thm}{Theorem}

\newtheorem{lemma}{Lemma}[section]
\newtheorem{prop}{Proposition}[section]

\theoremstyle{definition}

\newcommand{\be}{\begin{equation}}
\newcommand{\ee}{\end{equation}}

\newcommand{\lb}{\left}
\newcommand{\rb}{\right}

\newcommand{\mc}{\mathcal}

\newcommand{\ms}{\mathscr}
\newcommand{\mf}{\mathfrak}
\newcommand{\bb}{\mathbb}

\newcommand{\st}{~\mid~}

\newcommand{\eqsp}{\hspace{10pt};\hspace{10pt}} 

\newcommand{\union}{\cup} 
\newcommand{\inter}{\cap} 

\newcommand{\norm}[1]{\lb\Vert\, #1 \,\rb\Vert}		
\newcommand{\inp}[1]{\lb\langle #1 \rb\rangle}		

\newcommand{\surgrav}{\kappa}

\newcommand{\sqh}{\sqrt{h}}  

\newcommand{\Lie}{\pounds} 
\newcommand{\defn}{\mathrel{\mathop:}=} 
\newcommand{\svp}{\underline} 

\newcommand{\dom}{\mathrm{dom}} 

\begin{document}

\title{Black Hole Instabilities and Exponential Growth}
\author{Kartik Prabhu}\email{kartikp@uchicago.edu}
\author{Robert M. Wald}\email{rmwa@uchicago.edu}
\affiliation{Enrico Fermi Institute and Department of Physics,\\
The University of Chicago, Chicago, IL 60637, USA}

\begin{abstract}
Recently, a general analysis has been given of the stability with respect to axisymmetric perturbations of stationary-axisymmetric black holes and black branes in vacuum general relativity in arbitrary dimensions.
It was shown that positivity of canonical energy on an appropriate space of perturbations is necessary and sufficient for stability. However, the notions of both ``stability'' and ``instability'' in this result are significantly weaker than one would like to obtain. In particular, if there exists a perturbation with negative canonical energy, ``instability'' has been shown to occur only in the sense that this perturbation cannot asymptotically approach a stationary perturbation at late times. In this paper, we prove that if a perturbation of the form $\pounds_t \delta g$---with $\delta g$ a solution to the linearized Einstein equation---has negative canonical energy, then that perturbation must, in fact, grow exponentially in time. The key idea is to make use of the $t$- or ($t$-$\phi$)-reflection isometry, $i$, of the background spacetime and decompose the initial data for perturbations into their odd and even parts under $i$. We then write the canonical energy as \(\ms E\ = \ms K + \ms U\), where \(\ms K\) and \(\ms U\), respectively, denote the canonical energy of the odd part (``kinetic energy'') and even part (``potential energy''). One of the main results of this paper is the proof that $\ms K$ is positive definite for any black hole background. We use $\ms K$ to construct a Hilbert space $\ms H$ on which time evolution is given in terms of a self-adjoint operator $\tilde {\mc A}$, whose spectrum includes negative values if and only if $\ms U$ fails to be positive. Negative spectrum of $\tilde{\mc A}$ implies exponential growth of the perturbations in $\ms H$ that have nontrivial projection into the negative spectral subspace. This includes all perturbations of the form $\pounds_t \delta g$ with negative canonical energy. A ``Rayleigh-Ritz'' type of variational principle is derived, which can be used to obtain lower bounds on the rate of exponential growth.

\end{abstract}

\maketitle
\tableofcontents


\section{Introduction}\label{sec:intro}

Determining the stability of exact solutions to Einstein equations is a long standing problem. In order for a solution to be physically relevant, it is essential that sufficiently small perturbations not drive one away from that solution. However, the issue of full nonlinear stability is extremely difficult to analyze and has been settled only for Minkowski spacetime \cite{Ch-Kl}. As a first step, it is important to analyze the stability of solutions to linear perturbations, which solve the linearized Einstein equations (or in general, other fields that satisfy a linear dynamical equation).

A class of solutions of considerable interest are black hole and black brane spacetimes. All of the results of this paper will apply to the linear stability with respect to axisymmetric perturbations of stationary and axisymmetric black holes and black branes in vacuum general relativity in arbitrary dimensions. In order to keep our analysis and notation less cumbersome, we shall restrict consideration in this paper to the case of black holes, but the generalization of our results to the case of black branes is straightforward.

In \(4\)-spacetime-dimensions, the Schwarzschild and Kerr solutions have long been known to be the unique stationary black hole solutions to the the vacuum Einstein equation. These solutions are believed to physically describe the asymptotic final states of gravitational collapse. In order for this to be the case, it is essential that these solutions be stable. In higher dimensions, there exist many other types of black hole solutions, which are of interest for various theoretical reasons. It is of considerable interest to analyze the stability of these solutions as well.

To establish the linear stability of a solution, one must show that all suitably regular (e.g., smooth and satisfying appropriate asymptotic conditions) initial data for the linearized equations give rise to solutions that remain uniformly bounded in time. In the case of a linearly stable, stationary black hole solution, one would like to establish an even stronger result, namely that all suitably regular linearized solutions decay at asymptotically late times to a stationary solution. On the other hand, a considerably weaker notion of stability that is much easier to analyze is {\em mode stability}, i.e., the nonexistence of suitably regular solutions that grow exponentially in time. Clearly, mode instability implies instability in any other reasonable sense, but mode stability does not directly imply uniform boundedness or decay properties of linearized solutions.

As a simpler, model problem, it is useful to study scalar field perturbations of a black hole. For the case of a Schwarzschild black hole, it is easy to prove mode stability. Hilbert space methods were used in \cite{W-stab} to establish stability of Schwarzschild in the sense of uniform boundedness of solutions with a given spherical harmonic angular dependence. However, in addition to the restriction to a given spherical harmonic angular dependence, this approach required the imposition of the unwanted restriction of the vanishing of the perturbation on the bifurcation $2$-sphere. This latter restriction arose from the requirement that the perturbation lie in the Hilbert space used in the analysis and, as pointed out in \cite{Daf-Rod-lec}, it is equivalent to the requirement that the perturbation be expressible as $\pounds_t$ of another perturbation, where $\pounds_t$ denotes the Lie derivative with respect to the timelike Killing field, $t^\mu$, of the background spacetime. The restriction to a given spherical harmonic angular dependence can be straightforwardly removed and, indeed, the method can be generalized to analyze the stability of any static spacetime to scalar perturbations \cite{Wald-EYM}. However, it required a ``trick'' for Kay and Wald \cite{KW-stab} to remove the restriction of the vanishing of the perturbation on the bifurcation $2$-sphere, thereby establishing uniform boundedness of scalar perturbations of Schwarzschild. Nevertheless, the methods used in \cite{W-stab, Wald-EYM, KW-stab} are inadequate for establishing decay of perturbations. More recently, improved methods, described in \cite{Daf-Rod-lec}, have established decay of a scalar field in Schwarzschild and, indeed, similar results have been proven to hold in the much more difficult case of a Kerr black hole (including non-axisymmetric perturbations) \cite{FKSY, And, Tataru, DRS-stab}. However, these methods rely on very detailed properties of the Schwarzschild and Kerr metrics and cannot straightforwardly be generalized to arbitrary higher dimensional black holes.

For the case of gravitational perturbations, mode stability of $4$-dimensional Schwarzschild spacetime follows immediately from the form of the decoupled equations originally obtained by Regge and Wheeler \cite{Regge-Wheeler} and Zerilli \cite{Zerilli}. These results have been extended to prove mode stability of higher-dimensional Schwarzschild spacetimes by Ishibashi and Kodama \cite{IK-stab}. The methods of \cite{W-stab, KW-stab} could then be applied to show uniform boundedness of solutions with a given spherical harmonic angular dependence. However, on account of the gauge choices made, it is not straightforward to remove the restriction to a given spherical harmonic angular dependence. The decay results of \cite{Daf-Rod-lec} have not yet been generalized to gravitational perturbations, even for the case of Schwarzschild.

Recently, a new approach to investigating the linear stability of arbitrary static or stationary and axisymmetric black holes and black branes to axisymmetric gravitational perturbations was given in \cite{HW-stab}. A conserved, symmetric, gauge-invariant quadratic form on initial data for perturbations, called the \emph{canonical energy}, was introduced and shown to be given by the formula
\be
{\ms E} = \delta^2 M - \sum_\Lambda \Omega^\Lambda \delta^2 J_\Lambda - \frac{\surgrav}{2 \pi} \delta^2 A
\ee
where $M$ denotes the ADM mass, $J_\Lambda$ denotes the ADM angular momenta (with the sum over $\Lambda$ being taken over independent ``rotational planes''), and $A$ denotes the surface area of the horizon. 
For perturbations with $\delta M = \delta J_\Lambda = \delta P_i = \delta A = 0$ (where $P_i$ denotes the ADM linear momentum), it was shown that $\ms E$ is non-degenerate on the space of linearized solutions modulo gauge and modulo perturbations to other stationary and axisymmetric black holes. Consequently, either $\ms E$ is positive definite on this space or it can take negative values. Positive definiteness of $\ms E$ immediately establishes mode stability, since the existence of an exponentially growing solution is incompatible with the presence of a positive definite conserved norm. On the other hand, if $\ms E$ can take negative values, then a further argument establishes instability as follows: It was shown in \cite{HW-stab} that for axisymmetric perturbations, the flux of $\ms E$ through null infinity and through the black hole horizon must be positive. Thus, if the perturbation were to suitably approach a perturbation to another stationary black hole at asymptotically late times (with the limit suitably taken along the orbits of the timelike Killing field of the background), then ${\ms E}' < {\ms E}$, where ${\ms E}'$ denotes the canonical energy of this perturbation to another stationary black hole. Consequently, if $\ms E < 0$, then ${\ms E}' < 0$. Furthermore, for the limiting perturbation, we continue to have $\delta M = \delta J_\Lambda = 0$, since fluxes of mass and angular momenta are quadratic in the perturbation. But ${\ms E}' < 0$ then contradicts the fact that the canonical energy is degenerate on perturbations to other stationary black holes with $\delta M = \delta J_\Lambda = 0$, thus establishing ``instability'' in the sense that the perturbation cannot asymptotically approach a perturbation to another stationary black hole.

However, the analysis of \cite{HW-stab} leaves many questions unanswered with regard to the stability and instability of black holes to axisymmetric gravitational perturbations. In the first place, one would like to know---for any given black hole---whether the expression for $\ms E$ is positive. An explicit integral expression for $\ms E$ in terms of initial data $\lb( p_{ab}, q_{ab}\rb)$ for the linearized perturbation is given in Eq.86 of \cite{HW-stab}, where $q_{ab}$ denotes the perturbation of the spatial metric and $p_{ab}$ denotes the perturbation of the ADM momentum variable. Thus, the issue is simply to determine, for a given background, whether this expression is positive for all perturbations. However, the expression for $\ms E$ is quite complicated, and $\lb( p_{ab}, q_{ab}\rb)$ are not ``free'' but must satisfy constraints and boundary conditions. Furthermore, it would be difficult to show that the integral expression for $\ms E$ is positive without writing it in a form where the integrand is positive; however, although $\ms E$ is gauge invariant, the integrand is not, so this would require an appropriate choice of gauge and it is not obvious what gauge conditions will work. (The gauge conditions we will introduce later in this paper do not seem to work for this purpose.) For the case of a thermodynamically unstable black brane, there is a relatively obvious candidate perturbation that makes $\ms E$ negative, and failure of the positivity of $\ms E$ in this case was proven in \cite{HW-stab}. However, establishing the positivity of $\ms E$ in cases where it should be positive seems quite difficult. Indeed, even for the case of Schwarzschild, we have been able to prove positivity of $\ms E$ only by an indirect\footnote{In the case of Rindler spacetime, we have succeeded in proving positivity of $\ms E$ directly from the formula for $\ms E$.} argument---see remark 3 following \autoref{thm:positive-KE} of \autoref{sec:energies} below. We will not consider further in this paper the issue of determining the positivity of $\ms E$ in specific spacetimes.

As already mentioned above, positivity of $\ms E$ immediately implies mode stability, but if $\ms E$ is positive, one would expect significantly stronger results to hold, namely, the decay of perturbations. It is possible that the methods introduced in this paper may be adequate to prove uniform boundedness of perturbations when $\ms E$ is positive, i.e., the methods of this paper may be adequate for generalizing the results of Kay and Wald \cite{KW-stab} to axisymmetric gravitational perturbations of general stationary-axisymmetric black holes and black branes for which $\ms E$ is positive. However, considerable further analysis would be required to show this\footnote{The main things that would have to be shown are that (i) the operator $\mc A$ of \autoref{Adef} and its powers provide norms that, together with our gauge conditions, are equivalent to Sobolev norms and (ii) a version of the Kay and Wald ``trick'' can be used to eliminate the restriction to perturbations of the form of $\Lie_t$ applied to another perturbation.} and we shall not attempt to carry out this analysis here. Our methods would not, in any case, be adequate for proving decay results. To prove decay results, one would like to show that a suitably modified version of $\ms E$ is ``coercive.'' It would seem that a necessary first step toward showing this would be to have a good technique to show that $\ms E$ is positive (in cases where it is positive), but, as discussed in the previous paragraph, this is currently lacking. We shall not consider further in this paper this issue of obtaining strengthened stability results when $\ms E$ is positive.

As explained above, if $\ms E$ fails to be positive on a black hole or black brane spacetime, then it was shown in \cite{HW-stab} that one has instability in the sense that there exist perturbations that cannot approach a stationary perturbation at late times. However, if $\ms E$ fails to be positive, one might expect a much stronger result to hold, namely, the existence of initially well behaved perturbations that grow exponentially with time. The main purpose of this paper is to prove that this is indeed the case. Specifically, we will prove that if a perturbation of the form $\pounds_t \delta g$ has negative canonical energy, then that perturbation must, in fact, grow exponentially in time\footnote{In the last sentence of \cite{FMR}, the authors raise the question of whether violation of the local Penrose inequality implies the existence of perturbations that grow exponentially with time. Since it was shown in \cite{HW-stab} that violation of the Penrose inequality is equivalent to the failure of $\ms E$ to be positive, our results, in essence, answer that question in the affirmative.} (in all gauges). We now outline the key ideas used in the proof. 

The first key step---undertaken in \autoref{sec:lin-pert}---is to completely fix the gauge of our perturbations so as to obtain a unique time evolution. We consider a maximal slice $\Sigma$ of the background spacetime\footnote{Existence of $\Sigma$ is guaranteed by the results of \cite{Chr-Wald}.} and work in the space, $\ms P_{\rm gr}^\infty$, of linearized initial data $\delta X = (p_{ab}, q_{ab})$ on $\Sigma$ that lie in the intersection of the weighted Sobolev spaces defined by \autoref{eq:Sobolev-norm} below. Following the strategy of \cite{HW-stab} we impose the linearized constraints, the conditions $\delta M = \delta J_\Lambda = \delta P_i = 0$, and certain gauge conditions at the horizon by means of a projection map $\Pi_c$. We then fix the gauge completely by applying a similar projection map $\Pi_g$ that commutes with $\Pi_c$. The needed properties of $\Pi_c$ and $\Pi_g$ are proven in Appendix \ref{sec:proj-ops}. The space  $\Pi_c [\ms P_{\rm gr}^\infty]$ contains a gauge representative of any $\delta X \in \ms P_{\rm gr}^\infty$ that satisfies the linearized constraints and the conditions $\delta M = \delta J_\Lambda = 0$. The space $\ms V^\infty \defn \Pi_g \Pi_c [\ms P_{\rm gr}^\infty]$ contains a unique gauge representative of any such $\delta X$, and it therefore provides a suitable space to study dynamics.

The next key idea is to make use of the existence of a suitable reflection isometry, $i$, of the background spacetime. In the case of a static black hole, the maximal slice $\Sigma$ must be orthogonal to the static Killing field $t^\mu$. The desired isometry, $i$, can be constructed by ``$t$-reflection'' about $\Sigma$, i.e., $i$ is the diffeomorphism that takes a point $p$ lying at proper time $\tau$ along a normal geodesic $\gamma$ starting from $q \in \Sigma$ to the point $p'$ lying at proper time $-\tau$ along $\gamma$. 
For the case of a stationary-axisymmetric black hole, it has recently been proven \cite{SW-tphi} that if the stationary-axisymmetric isometries act trivially\footnote{By ``act trivially,'' we mean that if we remove from the spacetime manifold the points at which the stationary-axisymmetric Killing fields are linearly dependent, then the resulting manifold acquires the structure of a trivial principal fiber bundle with respect to the stationary-axisymmetric symmetries. This rules out behavior of axial Killing fields similar to that occurring in the Sorkin monopole \cite{Sor-monopole} \cite{GP-monopole}.}, then there exists a ``($t$-$\phi$)-reflection'' isometry, $i$, defined as follows\footnote{This result is well known in $4$-spacetime dimensions (see, e.g., \cite{Wald-book}), but, as explained in \cite{SW-tphi}, the proof in $4$-dimensions does not generalize to higher dimensions, so the additional assumption of a trivial action is needed in spacetime dimension greater than $4$.}: Let $\phi^\mu_1, \dots, \phi^\mu_p$ denote the collection\footnote{We may, if we wish, delete from this collection any axial Killing fields with vanishing horizon rotation provided that the action of the remaining Killing fields is also trivial \cite{SW-tphi}. In particular, for a static black hole, we may delete all of the axial Killing fields and take $i$ to be the ``$t$-reflection'' isometry, as assumed above.} of commuting axial Killing fields. Then, as shown in \cite{SW-tphi} the axial Killing fields are tangent to $\Sigma$ and are surface orthogonal within $\Sigma$, so there exists an isometry, $i_\Sigma:\Sigma \to \Sigma$, that reflects about a surface, $S$, in $\Sigma$. Now, extend $i_\Sigma$ to a spacetime difeomorphism $i$ by mapping a point $p$ lying at proper time $\tau$ along a normal geodesic starting from $q \in \Sigma$ to the point $p'$ lying at proper time $-\tau$ along the normal geodesic starting at $i_\Sigma(q)$. Then, as shown in \cite{SW-tphi}, $i$ preserves the initial data on $\Sigma$ and thus defines an isometry. We restrict consideration in the remainder of this paper to black holes that possess a $t$- or ($t$-$\phi$)-reflection isometry $i$.

The next key step---undertaken in \autoref{sec:energies}---is to decompose linearized initial data $\delta X = (p_{ab}, q_{ab}) \in \ms V^\infty$ into its odd and even parts under the action of $i$, i.e., we decompose the space, $\ms V^\infty$, of allowed initial data for perturbations as ${\ms V^\infty} = \ms V^\infty_{\rm odd} \oplus \ms V^\infty_{\rm even}$. We denote elements of $\ms V^\infty_{\rm odd}$ by ``$P$'' and we denote elements of $\ms V^\infty_{\rm even}$ by ``$Q$,'' so, by construction, we have $i^* P = - P$ and $i^* Q = Q$, where $i^*$ denotes the action induced by $i$. In the static case, we have $P = \lb( p_{ab}, 0\rb)$ and $Q = \lb(0, q_{ab}\rb)$. In the stationary-axisymmetric case, $P$ is a linear combination of the ``polar'' part of $p_{ab}$ and the ``axial'' part of $q_{ab}$, whereas $Q$ is a linear combination of the ``polar'' part of $q_{ab}$ and the ``axial'' part of $p_{ab}$ (see \autoref{eq:tphi-decomp} below). The linearized constraint equations are automatically invariant under $i$. Our gauge and boundary conditions are also chosen to be invariant under $i$, so $P$ and $Q$ decouple and we may treat them as independent perturbations. The decomposition $\ms V^\infty = \ms V^\infty_{\rm odd} \oplus \ms V^\infty_{\rm even}$ allows us to break up the canonical energy $\ms E$ into a sum of ``kinetic'' and ``potential'' energy,
\be
	{\ms E} = {\ms K} + {\ms U}
\ee
where $\ms K$ is the restriction of $\ms E$ to $\ms V^\infty_{\rm odd}$ and $\ms U$ is the  restriction of $\ms E$ to $\ms V^\infty_{\rm even}$; no ``($P$-$Q$)-cross-terms'' can arise on account of the reflection symmetry. Note that $\ms K$ defines a symmetric quadratic form on $\ms V^\infty_{\rm odd}$ and $\ms U$ defines a symmetric quadratic form on $\ms V^\infty_{\rm even}$.

The next key result---obtained at the end of \autoref{sec:energies}---is the proof that $\ms K$ is positive definite on $\ms V^\infty_{\rm odd}$, i.e.,
\be
{\ms K}(P,P) \geq 0 \, ,
\ee
with equality holding only if $P=0$. In particular, this shows that any failure of positivity of $\ms E$ must arise from the ``potential energy,'' $\ms U$. 

Time evolution is considered in \autoref{sec:evolution}. The ADM evolution equations mix the odd and even parts of the perturbation. Specifically, for axisymmetric perturbations, the evolution equations take the general form
\begin{eqnarray}
\dot{Q} = {\mc K} P \label{Qevolve} \\
\dot{P} = - {\mc U} Q \label{Pevolve}
\end{eqnarray}
for some operators $\mc K: \ms V^\infty_{\rm odd} \to \ms V^\infty_{\rm even}$ and $\mc U: \ms V^\infty_{\rm even} \to \ms V^\infty_{\rm odd}$, where the overdot denotes the Lie derivative, $\pounds_t$, with respect to $t^\mu$. Furthermore, the operators $\mc K$ and $\mc U$ appearing in the evolution equations are related to the above kinetic and potential energy quadratic forms by 
\begin{eqnarray}
{\ms K} (\tilde{P}, P) = \Omega_\Sigma (\tilde{P}, {\mc K} P)   \\
{\ms U} (\tilde{Q}, Q) = -\Omega_\Sigma (\tilde{Q}, {\mc U} Q)
\end{eqnarray}
where $\Omega_\Sigma$ denotes the symplectic form. These equations express the fact that the canonical energy $\ms E$ is actually a Hamiltonian for the linearized system. 

As we show in \autoref{sec:evolution}, the positivity of $\ms K$ allows us to define a ``${\mc K}^{-1}$ Hilbert space,'' $\ms H$, as follows. Consider $Q$ of the form $Q = {\mc K} P$ for some $P$; $Q$ will be of this form if and only if the initial data $P'=0$, $Q'=Q$ corresponds to a linearized solution that can be expressed as $\pounds_t$ of another linearized solution. Define the inner product of $\tilde{Q} = {\mc K} \tilde{P}$ and $Q = {\mc K} P$ by
\be
\inp{ \tilde{Q}, Q }_{\ms H} \defn {\ms K}(\tilde{P}, P)
\label{Hnorm}
\ee
and define $\ms H$ to be the Hilbert space completion of this inner product space. 

Combining the two time evolution equations \autoref{Qevolve} and \autoref{Pevolve}, we obtain the following second order evolution equation involving $Q$ alone:
\be\label{evoleq}
\ddot{Q} = -{\mc A} Q 
\ee
where
\be
{\mc A} \defn {\mc K} {\mc U} 
\ee
is a symmetric operator on the Hilbert space $\ms H$.
Consequently, by passing to a self-adjoint extension of $\mc A$, we can solve \autoref{evoleq} by spectral methods. For initial data in $\ms V^\infty_{\rm even} \cap \ms H$, we show that the Hilbert space solution must coincide with the PDE solution. It then follows that if there exists initial data of the form $Q = {\mc K} P$ for $P \in \ms V_{\rm odd}^\infty$ for which $\inp{ Q, {\mc A} Q }_{\ms H} = {\ms U} (Q,Q) < 0$, then the $\ms H$-norm of $Q$ must grow exponentially with time\footnote{The earliest reference that we are aware of for the basic argument that an equation of the form of \autoref{evoleq} will have exponentially growing solutions if $\mc A$ has negative spectrum is \cite{LMP}.}. The gauge invariant quantities $\ms K$ and $\ms U$ for this perturbation must also blow up exponentially, thus showing that the exponential blow up is not a gauge artifact.

The above instability result---including a quantitative statement about the rate of exponential growth---can be formulated as a ``Rayleigh-Ritz variational principle'' as follows. Let $P \in \ms V^\infty_{\rm odd}$ be any reflection-odd smooth initial data satisfying the constraints and such that $\delta J_\Lambda = \delta P_i = 0$. Write
\be
\omega^2 = \frac{{\ms U}({\mc K} P, {\mc K} P)}{{\ms K} (P,P)}
\ee
Then if $\omega^2 < 0$, the solution generated by $(P, Q=0)$ will grow at least as fast as $\exp(|\omega| t)$ (in all gauges). \\

The remainder of this paper is devoted to fleshing out these arguments. \autoref{sec:setup} discusses the properties of the background black hole spacetime and introduces the ADM formalism that we use throughout. In \autoref{sec:lin-pert} we introduce the spaces of interest for linearized perturbations and define the projection maps $\Pi_c$ and $\Pi_g$, which completely gauge-fix the perturbations. \autoref{sec:energies} introduces the splitting of the canonical energy into kinetic and potential energies and ends with the theorem showing that the kinetic energy is positive-definite. In \autoref{sec:evolution} we formulate the ADM equations as dynamical evolution equations on the Hilbert space $\ms H$. In \autoref{sec:exp} we solve these equations by spectral methods and prove that if the potential energy can be made negative on elements of $\ms H$, then there exist linear perturbations which have exponential growth in time. We also obtain our variational principle formulation. Appendix \ref{sec:proj-ops} provides the details of the construction of the projection operators $\Pi_c$ and $\Pi_g$. Appendix \ref{sec:other-fields} provides a parallel analysis for the Klein-Gordon scalar field and the electromagnetic field on a fixed black hole background. \\

In this paper, lower case Greek indices will be used to denote tensors on spacetime, e.g., $t^\mu$ denotes the timelike Killing field of the background black hole. Lower case Latin indices will be used to denote tensors on the initial data surface $\Sigma$, e.g., $q_{ab}$ denotes the perturbed metric on $\Sigma$. Upper case Latin indices will be used to denote tensors on the bifurcation surface, $B$, of the black hole, e.g., $\xi^A$ denotes a vector field on $B$. We will use the index ``$r$'' to denote projections normal to $B$, e.g., if $v_a$ is a one-form field on $\Sigma$, the normal component of its restriction to $B$ will be denoted $v_r = r^a v_a$, where $r^a$ is the unit normal to $B$. The spacetime derivative operator of the background black hole will be denoted as $\nabla_\mu$; the background derivative operator on $\Sigma$ will be denoted as $D_a$; the background derivative operator on $B$ will be denoted as $\ms D_A$. If we consider the restriction to $\Sigma$ of a spacetime vector field $\xi^\mu$, it will often be useful to represent $\xi^\mu$ as the pair $(\xi, \xi^a)$ where \(\xi \defn - u_\mu \xi^\mu \) (with $u^\mu$ the unit normal to $\Sigma$) is the component of $\xi^\mu$ normal to \(\Sigma\) while \(\xi^a\) is the projection of $\xi^\mu$ tangent to $\Sigma$. We will use the notation $\svp \xi \equiv (\xi, \xi^a)$ when we wish to view $\xi^\mu$ in this way.

\section{Background Spacetime}\label{sec:setup}

We consider a \((d+1)\)-spacetime-dimensional, asymptotically flat, static or stationary-axisymmetric black hole spacetime \(( M, g)\) shown in \autoref{fig:spacetime}, with a bifurcate Killing horizon \(  H \defn  H^+ \union  H^- \), and bifurcation surface \(  B \defn  H^+ \inter  H^- \). Let $t^\mu$ denote the time translation Killing field, i.e., the Killing field that is timelike near infinity. Let $\Sigma$ be an asymptotically flat Cauchy surface for one of the exterior wedges, smoothly terminating at $B$. Below, we will choose $\Sigma$ to be a maximal slice but we need not make this choice now. We also assume that the bifurcation surface \(B\) is compact, but we do not assume any further restrictions on its topology. Let  \( \Sigma_t \) denote the foliation obtained by applying time translations to \(\Sigma\). Let $u^\mu$ denote the future-directed unit normal to \(\Sigma\). We decompose $t^\mu$ into its normal and tangential parts relative to $\Sigma$, referred to as the \emph{lapse}, \(N = - u_\mu t^\mu\), and \emph{shift}, \(N^a\), on \(\Sigma\).

	\begin{figure}[h!]
		\centering
		\includegraphics[width=0.65\textwidth]{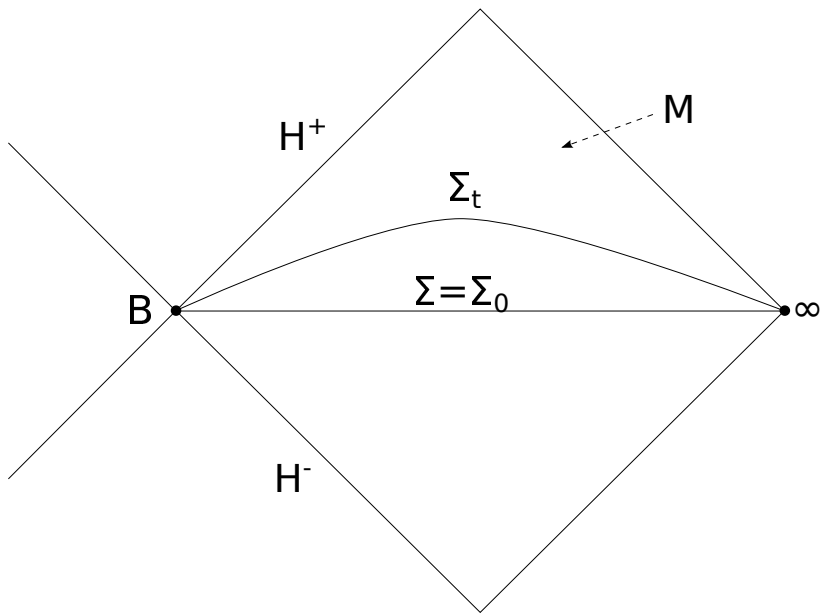}
		\caption{Carter-Penrose diagram of the black hole spacetime \((M,g)\).}\label{fig:spacetime}
	\end{figure}

Let \( h_{ab} \) denote the induced metric on $\Sigma$, with determinant with respect to some fixed volume form denoted by \(h\), and let \(K_{ab}\) denote the extrinsic curvature of $\Sigma$. The phase space of general relativity is the set of initial data on \(\Sigma\) denoted by \( X = \lb( \pi^{ab}, h_{ab}\rb) \) where \(\pi^{ab} \defn \sqh(K^{ab}-K~h^{ab}) \) is the canonical momentum-density conjugate to \(h_{ab}\) with \(K \defn K_{ab}h^{ab}\). In order to correspond to a solution to Einstein's equation, the initial data must satisfy the constraint equations \( {\mf C}^\mu= 0\), with
\begin{subequations}\label{eq:constraints}
	\begin{align}
		\mf C  & = - R + \tfrac{1}{h} \lb( \pi^{ab}\pi_{ab} - \tfrac{1}{d-1}\pi^2\rb) \label{eq:contraint-H}	\\
		\mf C^a & = -2D_b\lb( \frac{\pi^{ab}}{\sqh} \rb)  \label{eq:constraint-diff}
	\end{align}
\end{subequations}
where \(\pi \defn \pi_{ab}h^{ab}\), \( D_a \) denotes the covariant derivative compatible with \( h_{ab} \), and \( R_{ab} \) is the Ricci curvature of \(h_{ab}\). 

On the bifurcation surface \( B\), we introduce a unit normal vector \(r^a\) (pointing into \(\Sigma\)) and, without loss of generality, extend \(r^a\) in a neighbourhood of \( B\) such that it is a geodesic, so that, in particular, \((r^aD_ar_b) \vert_B= 0\). The projector onto \( B\) is \(s_{ab} \defn h_{ab} - r_ar_b\). As noted at the end of the introduction, we will use capital Latin letters \(A,B,C,\ldots\) to denote tensors on \( B\) and an index \(r\) to denote projections normal to \( B\). The induced metric on $B$ will be denoted as $s_{AB}$. We denote the metric-compatible covariant derivative operator on \( B\) by \(\ms D_A\), and we write \( d_r \defn r^aD_a \). The induced volume element on $B$ will be denoted as $\varepsilon_{{A_1} \dots A_{d-1}}$. Note that $N\vert_B =0$ and $(r^a D_a N)\vert_B =\surgrav$, where $\surgrav$ is the surface gravity of the black hole (assuming that the axial Killing fields are tangent to \(\Sigma\)). We also have \((\vartheta_\pm)\vert_B = 0 \), where
\be\label{eq:area-expansion-bg}
	\vartheta_\pm \defn s^{ab}\lb( K_{ab} \pm D_ar_b \rb) \, ,
\ee
i.e., $(\vartheta_+)\vert_B$ is the expansion of the outgoing null geodesics at $B$ and $(\vartheta_-)\vert_B$ is the expansion of the ingoing null geodesics at $B$.

The asymptotic flatness conditions on our stationary black hole are that there exist coordinates \( (x_1, \ldots, x_d) \) on \(\Sigma\) such that
\be\label{eq:asymp-bg}\begin{split}
	h_{ab} \sim \delta_{ab} + O(1/\rho^{d-2}) \eqsp \pi^{ab} \sim O(1/\rho^{d-1}) \\
	N \sim 1 + O(1/\rho^{d-2}) \eqsp N^a \sim O(1/\rho^{d-2})
\end{split}\ee
where \( \rho = \lb( x_1^2 + \ldots + x_d^2 \rb)^{1/2} \) near infinity. In addition,
\(k^{th}\) derivatives of the above quantities are required to fall off faster by an additional factor of \(1/\rho^k\). The asymptotic conditions on the lapse and shift ensure, in particular, that $t^\mu$ goes to an asymptotic time-translation at infinity.

The ADM time evolution equations are (see, e.g., Sec.E.2 of \cite{Wald-book} for \(d=3\) and Sec.VI.6 of \cite{CB-book} for general \(d\)):
\begin{subequations}\label{eq:evol-bg}
	\begin{align}
		\begin{split}		
		\tfrac{1}{\sqh}\dot\pi^{ab}   = & -N \lb( R^{ab} - \tfrac{1}{2} R h^{ab}\rb) + \tfrac{N}{2h}~h^{ab} \lb( \pi^{cd}\pi_{cd} - \tfrac{1}{d-1}\pi^2 \rb) - \tfrac{2N}{h}\lb( \pi^{ac} {\pi_c}^b - \tfrac{1}{d-1} \pi\pi^{ab} \rb) \\
		& +  D^aD^bN - h^{ab}\triangle N  + D_c \lb( N^c \tfrac{\pi^{ab}}{\sqh} \rb) - \tfrac{2}{\sqh}~ \pi^{c(a}D_cN^{b)} \\
		\end{split} \\ 
		&\dot h_{ab}  =  \frac{2N}{\sqh}\lb( \pi_{ab} - \tfrac{1}{d-1}\pi h_{ab} \rb) + 2 D_{(a}N_{b)}
	\end{align}
\end{subequations}
where the overdot denotes $\Lie_t$ and \(\triangle \defn D^aD_a\) is the Laplacian on \(\Sigma\). 
Since we are considering stationary black hole spacetimes, the left side of \autoref{eq:evol-bg} vanishes. 

We can significantly further simplify the right side of \autoref{eq:evol-bg} by choosing $\Sigma$ to be a maximal slice, whose existence\footnote{The proof of \cite{Chr-Wald} was given in the case of $4$ spacetime dimensions, but it generalizes straightforwardly to arbitrary dimensions.} was proven in \cite{Chr-Wald}. As shown in \cite{SW-tphi}, the axial Killing fields $\phi^\mu_\Lambda$, with $\Lambda = 1, \dots, p$, are then tangent to $\Sigma$, so we also may denote them as $\phi^a_\Lambda$. We will assume the existence of a ($t$-$\phi$)-reflection isometry about $\Sigma$, which, as discussed in the introduction, has been proven to exist (see \cite{SW-tphi}) if the action of the stationary and axisymmetric isometries is trivial. At a point of $\Sigma$, the ($t$-$\phi$)-reflection isometry reverses the directions of $t^\mu$ and $\phi^\mu_\Lambda$ but leaves the subspace orthogonal to $t^\mu$ and $\phi^\mu_\Lambda$ invariant. Since the normal, $u^\mu$, to $\Sigma$ must reverse sign under the isometry, it must be expressed as a (variable) linear combination of $t^\mu$ and the axial Killing fields $\phi^\mu_1, \dots, \phi^\mu_p$
\be
u^\mu (x) = \alpha(x) [t^\mu - \sum_\Lambda \bar{N}^\Lambda (x) \phi^\mu_\Lambda] \, .
\ee
Thus, the shift vector takes the form 
\be
N^a = \sum_\Lambda \bar N^\Lambda\phi^a_\Lambda \, .
\ee
The extrinsic curvature of $\Sigma$ is purely ``axial'' (i.e., odd under $\phi$-reflection), so $\pi^{ab}$ takes the form
\be
\pi^{ab} = 2\sqh \sum_\Lambda \pi^{\Lambda(a}\phi^{b)}_\Lambda 
\ee
with $\pi^\Lambda_a\phi^a_\Theta = 0$ for all $\Lambda,\Theta$. 

The axial Killing fields comprise a vector space $\bb V$, and it is useful to think of the $\Lambda$-index in $\phi^\mu_\Lambda$ as an abstract index of $\bb V$ rather than a labeling index running from $1$ to $p$. At each $x \in \Sigma$ where the Killing fields are linearly independent, we can then define a positive definite inverse metric $\Phi_{\Lambda\Theta}(x)$ on $\bb V$ by
\be
\Phi_{\Lambda\Theta}(x) \defn h_{ab}(x) \phi^a_\Lambda\phi^b_\Theta  
\ee
We will use $\Phi_{\Lambda\Theta}$ and its inverse, $\Phi^{\Lambda\Theta}$ to lower and raise \(\bb V\)-indices. Note, however, that \(D_a\Phi_{\Lambda\Theta} \neq 0\) so, while \(\phi^a_\Lambda\) satisfies Killing's equation, \(\phi^{a\Lambda}\) does not. The $(d-p)$-dimensional surface orthogonality of \(\phi^a_\Lambda\) within $\Sigma$ together with Killing's equation implies that 
	\be\label{eq:Dphi}
		D_a\phi_{b\Lambda} = -  \Phi^{\Theta\Xi} \phi_{\Theta[a}D_{b]}\Phi_{\Lambda\Xi} \, .
	\ee

With the above choice of $\Sigma$, the ADM evolution equations \autoref{eq:evol-bg} reduce to:
\be\label{eq:Ric-id}
NR_{ab} = D_aD_b N -2 N\lb( \pi_a^\Lambda\pi_{b\Lambda} - \pi_c^\Lambda \pi^{c\Theta} \phi_{a\Lambda}\phi_{b\Theta} \rb)
\ee
and
\be\label{eq:shift-id}
D_a \bar N^\Lambda  = - 2 N\pi_a^\Lambda
\ee
In addition, the constraint equations \autoref{eq:constraints} become:
\be
R = 2\pi_a^\Lambda \pi^a_\Lambda
\ee
and
\be
D_a\pi^a_\Lambda = 0
\ee
These relations simplify considerably in the static case, where $\pi^{ab} = 0$ and $N^a = 0$. The evolution equations then reduce to 
\be
NR_{ab} = D_aD_bN
\ee
and the constraint equations reduce to
\be
R = 0 \, .
\ee


\section{Linear Perturbations: Constraints, Boundary Conditions, and Gauge Conditions}\label{sec:lin-pert}

Let \( X(\lambda) = \lb(\pi^{ab}(\lambda), h_{ab}(\lambda) \rb)\) be a one-parameter family of initial data that is jointly smooth in \(\lambda\) and point on $\Sigma$, with \( X(0) \) corresponding to initial data for a stationary black hole, as discussed in the previous section. Linearized perturbations off of \(X(0) \) are characterized by
\be
\delta X = (p_{ab},  q_{ab})
\ee
where
\be
\sqh p^{ab} \defn \delta\pi^{ab} = \frac{d}{d\lambda}\pi^{ab}(\lambda)\vert_{\lambda=0} \eqsp q_{ab} \defn \delta h_{ab} = \frac{d}{d\lambda}h_{ab}(\lambda)\vert_{\lambda=0}
\ee

The main task of this section is to define the space of perturbations, $\ms V^\infty$, that we will work with. We wish to define a space of perturbations that include all perturbations that (i) are smooth and satisfy appropriate asymptotic fall-off properties at infinity, (ii) satisfy the constraints, (iii) satisfy the horizon gauge conditions needed to define the canonical energy, and (iv) have vanishing perturbed ADM mass, linear momentum, and angular momentum as needed for the stability analysis. Furthermore, we want the perturbations to satisfy the property that (v) they are completely ``gauge fixed.''

The smoothness and asymptotic fall-off properties will be enforced by requiring the initial data $\delta X$ on $\Sigma$ to lie in suitable weighted Sobolev spaces. Let $\rho$ be a positive function that goes to $1$ in a neighborhood of $B$ and approaches \( \lb( x_1^2 + \ldots + x_d^2 \rb)^{1/2} \) near infinity. Let \(W^k_\rho\) denote the closure of data in $C^\infty(\Sigma)$ that vanish in a neighborhood of infinity in the norm
\be\label{eq:Sobolev-norm}
	\norm{\delta X}^2_{W^k_\rho} \defn \sum\limits_{n=0}^k \int_\Sigma \rho^{2n}\left[(D_{a_1}\ldots D_{a_n} p_{bc}) (D^{a_1}\ldots D^{a_n} p^{bc}) + (D_{a_1}\ldots D_{a_n} q_{bc}) (D^{a_1}\ldots D^{a_n} q^{bc})\right]
\ee
Here and below, the integral over \(\Sigma\) is taken with respect the volume form induced by the background metric \(h_{ab}\). The space of interest for our analysis is
\be
\ms P_{\rm gr}^\infty \defn \inter_k W^k_\rho \, .
\label{Pgrinfty}
\ee
The family of $W^k_\rho$-norms gives $\ms P_{\rm gr}^\infty$ the natural structure of a Fr\'{e}chet space. The finiteness of all the above Sobolev norms together with the standard Sobolev estimates, implies that all elements of $\ms P_{\rm gr}^\infty$ are smooth, with $p_{ab}, q_{ab} = o(1/\rho^{d/2})$ as $\rho \to \infty$ and all spatial derivatives falling faster by corresponding powers of $\rho$. 
Note that the fall-off conditions for data in $\ms P_{\rm gr}^\infty$ are weaker than normally assumed (see \autoref{eq:asymp-bg} above), except for $d=3,4$ where our weighted Sobolev space conditions impose a faster than normal fall-off in $q_{ab}$ and thereby exclude perturbations that change the ADM mass. Since we will be interested only in perturbations for which $\delta M = 0$, this will not impose any unwanted restrictions when $d=3,4$. Note also that we impose the same fall-off conditions on $p_{ab}$ and $q_{ab}$ rather than requiring faster fall-off on $p_{ab}$. It is important for our constructions below that we treat $p_{ab}$ and $q_{ab}$ on an ``equal footing.'' The imposition of the same fall-off rates on $p_{ab}$ and $q_{ab}$ would cause difficulties if we wished to consider time evolution to ``boosted slices.'' However, when we consider time evolution of perturbations in \autoref{sec:evolution}, we will evolve only to ``time translated'' Cauchy surfaces.

Although our interest is in perturbations $\delta X \in \ms P_{\rm gr}^\infty$, it is very convenient to perform constructions in the larger $L^2$-space 
\be
\ms P_{\rm gr} \defn W^0_\rho \supset \ms P_{\rm gr}^\infty
\ee
with inner product
\be\label{eq:inner-prod-defn}
	\inp{ \widetilde{\delta X} , \delta X } = \int_\Sigma \lb( \tilde p^{ab}p_{ab} + \tilde q^{ab}q_{ab} \rb) \, ,
\ee
since, as we shall see below, the constraints and gauge conditions can be expressed in terms of orthogonal projection maps on $\ms P_{\rm gr}$. Note that the symplectic form
\be\label{eq:symplectic-form}
\Omega_\Sigma \lb( X;\widetilde{\delta X},\delta X \rb) \defn \int_\Sigma \lb( \tilde p^{ab} q_{ab} - \tilde q_{ab}p^{ab} \rb)
\ee
is represented on $\ms P_{\rm gr}$ by the bounded linear map \(\mc{S}\) given by
\be\label{eq:symplectic-op-defn}
\mc S \lb( p_{ab}, q_{ab} \rb) = \lb( q_{ab}, -p_{ab} \rb) 
\ee
Note further that \(\mc S^* = - \mc S\) and \(\mc S^2 = -1\), so \(\mc S: \ms P_{\rm gr} \to \ms P_{\rm gr} \) is an orthogonal map.

Of course, we are interested only in the elements \(\delta X \in \ms P_{\rm gr}^\infty \) that satisfy the linearized constraints. We may view the operator obtained by linearizing the constraints \autoref{eq:constraints} off of the stationary black hole background as a map, $\mc L$, taking smooth initial data on $\Sigma$ to a pair 
\be
\svp\xi \equiv (\xi, \xi^a)
\ee
consisting of a smooth scalar field, $\xi$, and smooth vector field, $\xi^a$, on $\Sigma$:
\be\label{eq:L-defn}\begin{split}
	\mc L\begin{pmatrix} p_{ab} \\ q_{ab} \end{pmatrix} \defn \begin{pmatrix} \begin{split}
		&\tfrac{2}{\sqh}\lb( \pi^{ab} - \tfrac{1}{d-1}\pi~ h^{ab} \rb)p_{ab} + \tfrac{2}{h}\lb( \pi^{ac} {\pi_c}^b - \tfrac{1}{d-1} \pi\pi^{ab} \rb) q_{ab} \\
		& -\tfrac{1}{h}\lb[ \pi^{ab}\pi_{ab} - \tfrac{1}{d-1}\pi^2 \rb]q -  D^aD^bq_{ab} + \triangle q + R^{ab}q_{ab}
		\end{split} \\[25pt]
		-2D_bp^{ab} - \tfrac{1}{\sqh}\pi_{bc} \lb ( D^c q^{ba} + D^b q^{ca} - D^a q^{bc} \rb)
	\end{pmatrix}
\end{split}\ee
where \(q \defn q_{ab}h^{ab}\). In this notation, the linearized constraints are
\be\label{linconst}
\svp{\mf c} \defn \mc L \, \delta X = 0 \, .
\ee
The formal $L^2$-adjoint of $\mc L$ is the map
\be \label{eq:L*-defn}\begin{split}
	\mc L^*\begin{pmatrix} \xi \\ \xi^a \end{pmatrix} \defn \begin{pmatrix}
		\frac{2\xi}{\sqh}\lb( \pi_{ab} - \tfrac{1}{d-1}\pi h_{ab} \rb) + 2 D_{(a}\xi_{b)} \\[25pt]
	 \begin{split}
		& \xi \lb( R^{ab} - \tfrac{1}{2} R h^{ab}\rb) - \tfrac{\xi}{2h}~h^{ab} \lb( \pi^{cd}\pi_{cd} - \tfrac{1}{d-1}\pi^2 \rb) + \tfrac{2\xi}{h}\lb( \pi^{ac} {\pi_c}^b - \tfrac{1}{d-1} \pi\pi^{ab} \rb) \\
		& -  D^aD^b\xi + h^{ab}\triangle\xi  - D_c \lb( \xi^c \tfrac{\pi^{ab}}{\sqh} \rb) + \tfrac{2}{\sqh}~ \pi^{c(a}D_c\xi^{b)}
		\end{split}
	\end{pmatrix}
\end{split}\ee
where we have used the background constraints \autoref{eq:constraints} in computing the adjoint. Note that ${\mc S}^* {\mc L}^* \svp\xi$ corresponds precisely to the infinitesimal gauge transformation generated by the vector field $\xi^\mu \equiv (\xi, \xi^a)$. Note also that since all gauge transformations are solutions to the linearized constraints, the equation \(\mc L \mc S^* \mc L^* = 0 \) holds as an identity.

As explained in \cite{HW-stab}, in addition to the constraints, we also must impose some restrictions on the perturbations and some gauge conditions at $B$ in order that the canonical energy have suitable gauge invariance and non-degeneracy properties. The restrictions we impose are $\delta M = \delta J_\Lambda = 0$, where $J_\Lambda$ denotes the angular momenta conjugate to the axial Killing fields $\phi^\mu_\Lambda$; we also impose $\delta P_i = 0$ where $P_i$ denotes the linear momenta, but this is not a physical restriction since this condition can be achieved via an asymptotic Lorentz boost. The gauge conditions at $B$ that were imposed in \cite{HW-stab} were $\delta\varepsilon = 0$ and $\delta \vartheta_+ = 0$, where 
\be
\delta\varepsilon = \tfrac{1}{(d-1)!}\varepsilon^{A_1 \dots A_{d-1}} \delta\varepsilon_{A_1 \dots A_{d-1}}
\ee
is the perturbed area element\footnote{Since $\delta M = \delta J_\Lambda = 0$, the first law of black hole mechanics implies $\delta A=0$, in which case $\delta\varepsilon = 0$ can be imposed by a gauge choice.} of $B$ and $\delta \vartheta_+ $ is the perturbed outgoing expansion of $B$. However, for our purposes, it is essential that our gauge conditions at $B$ respect the ($t$-$\phi$)-reflection isometry $i$. For this reason, we impose the additional gauge condition that the ingoing expansion of $B$ also vanish, $\delta \vartheta_- = 0$. That both $\delta \vartheta_+ = 0$ and $\delta \vartheta_- = 0$ can be achieved without imposing any physical restrictions on the perturbation can be seen as follows: It was proven in \cite{HW-stab} that the gauge condition $\delta \vartheta_+ = 0$ always can be imposed. However, the condition $\delta \vartheta_+ = 0$ does not uniquely determine a $2$-surface; rather it holds on all cross-sections of an outgoing null hypersurface in the perturbed spacetime. The condition $\delta \vartheta_- = 0$ similarly holds on all cross-sections of an ingoing null hypersurface. The intersection of these hypersurfaces defines a unique surface on which both $\delta \vartheta_+ = 0$ and $\delta \vartheta_- = 0$. An infinitesimal differmorphism that moves this surface to $B$ achieves the desired gauge condition.

Thus, the conditions we impose on perturbations are
\be\label{admcond}
\delta M = \delta J_\Lambda = \delta P_i = 0
\ee
\be\label{horcond}
\delta \epsilon = \delta \vartheta_+ = \delta \vartheta_- = 0 \, .
\ee
As already explained above, only the conditions $\delta M = \delta J_\Lambda = 0$ are physical restrictions on the perturbations; the remaining conditions $\delta P_i = \delta \epsilon = \delta \vartheta_+ = \delta \vartheta_- = 0$ can be achieved by a choice of gauge. The relations \autoref{horcond} can be written more explicitly in terms of the perturbed initial data as
\begin{subequations}\label{eq:bound-defn}\begin{align}
	0 = \delta\varepsilon & = \tfrac{1}{2}{q^A}_A  \label{eq:bound-defn-area} \\ 
	0 = \delta\vartheta_{\rm odd} & = -p_{rr}  \label{eq:bound-defn-expodd} \\
	0 = \delta\vartheta_{\rm even} & = \tfrac{1}{2}\lb( d_r {q^A}_A - 2 \ms D^Aq_{Ar} \rb) \label{eq:bound-defn-expeven}
\end{align}\end{subequations}
where \(\delta\vartheta_{\pm} = \delta\vartheta_{\rm odd} \pm \delta\vartheta_{\rm even}\). Here we have used the subscripts ``\(\rm{odd}\)'' and ``\(\rm{even}\)'' to denote parts of the perturbed expansion which are odd and even under the \(t\)-\(\phi\) reflection isometry, as will be discussed further in \autoref{sec:energies}.

Following the strategy of \cite{HW-stab}, we impose the linearized constraints \autoref{linconst} and the additional conditions \autoref{admcond} and \autoref{horcond} by the following procedure. Let $\ms W_c$ be the subspace of \(\ms P^\infty_{\rm gr}\) comprised by smooth gauge transformations that become an asymptotic translation or a rotation with respect to an axial Killing field\footnote{Note that {\em all} axial Killing fields of the background spacetime are included in the definition of $\ms W_c$. In particular, for perturbations of Schwarzschild, $\xi^\mu$ may approach an arbitrary asymptotic rotation at infinity.} at infinity and whose projection onto $B$ vanishes
\be\label{eq:Wc-defn}\begin{split}
	\ms W_c  & \defn \{ \mc S^*\mc L^* \svp\xi \in \ms P_{\rm gr}^\infty \st \xi^A \vert_B= 0 \text{ and } \xi^\mu \text{ goes to an asymptotic translation plus } \\
 	&\qquad \text{ rotation with respect to any axial Killing field at infinity} \}
\end{split}\ee
where, as above, $\svp \xi$ denotes the pair consisting of a smooth scalar field, $\xi$, and smooth vector field, $\xi^a$, on $\Sigma$. For any $\svp \xi$ as in \autoref{eq:Wc-defn} and any $\delta X \in \ms P_{\rm gr}^\infty$, a direct computation yields \cite{HW-stab}
\be\label{eq:bound-integral}
	\inp{\mc L^*\svp\xi, \delta X} = \inp{\svp\xi, \mc L\delta X} + \int_B \lb[ -2d_r \xi~ \delta\varepsilon + \xi~(\delta\vartheta_+  + \delta\vartheta_-) + \xi_r~(\delta\vartheta_+  - \delta\vartheta_-) \rb] + \delta H_{\svp\xi} 
\ee
where \(\delta H_{\svp\xi}\) is the boundary term at infinity that represents the perturbed ADM conserved quantity corresponding to the asymptotic symmetry $\xi^\mu$. Thus, $\delta X \in \ms P_{\rm gr}^\infty$ satisfies the constraints \autoref{linconst} and the conditions \autoref{admcond} and \autoref{horcond} if and only if it is $L^2$-orthogonal to $\mc L^* \svp\xi$ for all $\svp \xi$ as in \autoref{eq:Wc-defn}. 

Let $\ms V_c $ denote the closed subspace of $\ms P_{\rm gr}$ that is symplectically-orthogonal to \(\ms W_c \),
\be\label{eq:Vc-defn}
	\ms V_c  \defn \ms W_c^{\mc S\perp} = \{ \delta X \in \ms P_{\rm gr} \st \inp{\widetilde{\delta X}, \mc S \delta X } = 0 \text{ for all } \widetilde{\delta X} \in \ms W_c \}
\ee
Then elements of $\ms V_c$ {\em weakly} satisfy the constraints \autoref{linconst} and the conditions \autoref{admcond} and \autoref{horcond}. The perturbations of interest for our analysis are those that lie in the subspace
\be
\ms V_c^\infty \defn \ms V_c \cap \ms P_{\rm gr}^\infty \, .
\ee
Let $\Pi_c : \ms P_{\rm gr} \to \ms P_{\rm gr}$ be the orthogonal projection operator onto $\ms V_c$. By slight modifications of the proof of Lemma 3 of \cite{HW-stab}, we show in Appendix \ref{sec:proj-ops} that---by virtue of the fact that $\Pi_c$ is constructed by solving an elliptic system---we have
\be
\ms V_c^\infty  = \Pi_c  [\ms P_{\rm gr}^\infty]  \, .
\ee
Furthermore, we show that $\Pi_c: \ms P_{\rm gr}^\infty \to \ms V_c^\infty$ is continuous in the natural (Fr\'{e}chet) topology of $\ms P_{\rm gr}^\infty$. Note also that since $\ms P_{\rm gr}^\infty$ is dense in $\ms P_{\rm gr}$ in the $L^2$-topology and $\Pi_c$ is an orthogonal projection map, we have
\be
\overline{\ms V_c^\infty} = \ms V_c
\label{dense}
\ee
where $\overline{\ms V_c^\infty}$ denotes the closure of $\ms V_c^\infty$ in the $L^2$-topology.

The initial data in $\ms V_c^\infty$ has considerable gauge freedom, as only the gauge conditions \autoref{horcond} at $B$ have been imposed, as well as $\delta P_i = 0$. Thus all smooth gauge transformations are allowed that preserve these conditions at $B$ and approach asymptotic translations and/or rotations at infinity. When we consider time evolution in \autoref{sec:evolution}, this gauge freedom would create a nuisance for defining time evolution operators. Fortunately, it is possible to fix the gauge freedom completely as follows. The allowed gauge transformations $\hat{\delta} X \defn \mc S^*\mc L^* \svp\xi$ are precisely the ones lying in $\ms P_{\rm gr}^\infty$ that additionally satisfy \autoref{admcond} and \autoref{horcond}. Satisfaction of \autoref{admcond} requires that $\xi^\mu$ asymptotically approach a translation or rotation\footnote{Since $M \neq 0$ in the background, an arbitrary asymptotic Lorentz boost will yield $\delta P_i \neq 0$. However, an arbitrary translation or rotation will not change $M$, $J_\Lambda$, or $P_i$ to first order. Note that here the rotation may be arbitrary, i.e., it does not have to be along $\phi^a_\Lambda$.} at infinity. From \autoref{eq:bound-defn}, it follows that satisfaction of \autoref{horcond} requires
\begin{subequations}\label{eq:bound-gauge}\begin{align}
	0 = \hat{\delta}\varepsilon & = \ms D_A\xi^A \label{eq:bound-gauge-area} \\	
	0 = \hat{\delta}\vartheta_{\rm odd} &  = \lb( \ms D^A\ms D_A - 2\pi_r^\Lambda\phi^A_\Lambda\ms D_A + \pi_r^\Lambda \pi_{r\Lambda} - \tfrac{1}{2}\ms R \rb)\xi \label{eq:bound-gauge-expodd} \\
	0 = \hat{\delta}\vartheta_{\rm even} & = -\lb( \ms D^A\ms D_A - 2\pi_r^\Lambda\phi^A_\Lambda\ms D_A + \pi_r^\Lambda \pi_{r\Lambda} - \tfrac{1}{2}\ms R \rb)\xi_r 	\label{eq:bound-gauge-expeven}
\end{align}\end{subequations}
where \(\ms R\) is the Ricci scalar of \(B\). The requirements that $\hat{\delta}\vartheta_{\rm odd} = \hat{\delta}\vartheta_{\rm even} = 0$ can be shown to imply $\xi\vert_B = \xi_r\vert_B = 0$ by the same argument as given in \cite{HW-stab} (with their \(\beta^A \equiv 2\sqh \pi^A_r = 2\pi_r^\Lambda\phi^A_\Lambda \) and we have used the axisymmetry of the background only), which extended arguments given in  \cite{And-Mars-Sim} and \cite{Gal-Sch}. Thus, the space $\ms W_g$ of allowed gauge transformations is given by
\be\label{eq:Wg-defn}\begin{split}
	\ms W_g   \defn &\{ \mc S^*\mc L^* \svp\xi \in \ms P_{\rm gr}^\infty \st \xi \vert_B = \xi_r\vert_B  = \ms D_A\xi^A \vert_B = 0  \text{ and } \xi^\mu \text{ asymptotically } \\
 	& \qquad \text{ approaches an arbitrary translation plus rotation at infinity} \} 
\end{split}\ee
Remarkably, $\ms W_g$ differs from $\ms W_c$ only in that the condition $\xi^A \vert_B = 0 $ in the definition of $\ms W_c$ has been replaced by the conditions $\xi \vert_B = \xi_r\vert_B  = \ms D_A\xi^A \vert_B = 0$, and the asymptotic conditions at infinity are somewhat different.

We now fix the gauge completely on $\ms V_c$ by requiring orthogonality to \( \ms W_g \) in the \(L^2\)-inner product. For smooth elements $\delta X \in \ms V^\infty_c$, this $L^2$-orthogonality will hold if and only if 
\begin{subequations}\label{eq:gauge-conds}\begin{align}
	\svp{\mf g} \defn \mc L\mc S\delta X = 0  \label{eq:gauge-conds-bulk} \\
	\begin{pmatrix} \delta M,  \delta P,  \delta J \end{pmatrix}(\mc S\delta X) = 0 \label{eq:gauge-conds-adm} \\
	\begin{pmatrix} \delta\varepsilon,  \delta\varpi_{AB}\end{pmatrix}(\mc S\delta X)\vert_B = 0 \label{eq:gauge-conds-B}
\end{align}\end{subequations}
where
	\be\label{eq:varpi-defn}
		\delta\varpi_{AB} \defn \ms D_{[A}\lb( p_{B]r} + \pi_r^\Lambda\phi^C_\Lambda q_{B]C} \rb) \, ,
	\ee
and ``$\delta J$'' in \autoref{eq:gauge-conds-adm} includes all of the angular momenta, not just the ones conjugate to $\phi^a_\Lambda$. The conditions at $B$ expressed in \autoref{eq:gauge-conds-B} take the explicit form
\begin{subequations}\label{eq:gauge-conds-B-full}\begin{align}
	0 = \delta\varepsilon(\mc S\delta X)\vert_B & = -\tfrac{1}{2}{p^A}_A \label{eq:gauge-conds-B-area} \\
	0 = \delta\varpi_{AB} (\mc S\delta X)\vert_B & = \ms D_{[A}\lb( q_{B]r} - \pi_r^\Lambda \phi^C_\Lambda p_{B]C} \rb) \label{eq:gauge-conds-B-varpi}
\end{align}\end{subequations}
Thus, our gauge conditions on $\delta X$ correspond to the conditions \autoref{linconst}, \autoref{admcond}, and \autoref{eq:bound-defn} with $p_{ab} \to q_{ab}$ and $q_{ab} \to - p_{ab}$ {\it except} that the conditions requiring the vanishing of the expansions, $\delta \vartheta_{\rm odd}$ and $\delta \vartheta_{\rm even}$, in \autoref{eq:bound-defn-expodd} and \autoref{eq:bound-defn-expeven} are deleted and replaced by \autoref{eq:gauge-conds-B-varpi}, and the condition $\delta J_\Lambda = 0$ is replaced by $\delta J = 0$.

Let $\ms V_g$ denote the subspace of $\ms P_{\rm gr}$ that is orthogonal to $\ms W_g$. Then elements of $\ms V_g$ weakly satisfy our gauge conditions \autoref{eq:gauge-conds}. Let $\Pi_g: \ms P_{\rm gr} \to \ms P_{\rm gr}$ denote the orthogonal projection operator onto $\ms V_g$. It is shown in Appendix \ref{sec:proj-ops} that, as for $\Pi_c$, we have
\be
\ms V_g^\infty  \defn \ms V_g \cap \ms P_{\rm gr}^\infty = \Pi_g  [\ms P_{\rm gr}^\infty] \, .
\ee
Furthermore, as for $\Pi_c$, we have that $\Pi_g: \ms P_{\rm gr}^\infty \to \ms V_g^\infty$ is continuous in the (Fr\'{e}chet) topology of $\ms P_{\rm gr}^\infty$. 

Now, by construction, we have $\ms W_g \subset \ms V_c^\infty$, since the gauge transformations in $\ms W_g$ were chosen to be precisely the gauge transformations in $\ms P_{\rm gr}^\infty$ that lie in $\ms V_c^\infty$. Taking perp-spaces, we find that $\ms V_g \supset (\ms V_c^\infty)^\perp = (\ms V_c)^\perp$, where the last equality follows from \autoref{dense}. This implies that $(I-\Pi_c) = \Pi_g (I-\Pi_c) = (I-\Pi_c) \Pi_g$, which, in turn, implies that $\Pi_c$ and $\Pi_g$ commute
\be
\Pi_c \Pi_g = \Pi_g \Pi_c \, .
\ee

The space of interest for us is
\be
\ms V^\infty \defn \Pi_g \Pi_c [\ms P_{\rm gr}^\infty] = \ms V_g^\infty \cap \ms V_c^\infty \, .
\ee
This space contains a unique gauge representative of every element of $\ms V_c^\infty$. In the remainder of this paper, we will analyze dynamical stability for initial data in the space $\ms V^\infty$.


\section{Kinetic and Potential Energies; Positivity of Kinetic Energy}\label{sec:energies}

As discussed in the Introduction, we can use the ($t$-$\phi$)-reflection isometry, $i$, of the background stationary black hole to decompose an arbitrary perturbation into its ``odd'' and ``even'' parts, $P$ and $Q$, under the action of $i$. Let $\Sigma$ be a ($t$-$\phi$)-reflection symmetric Cauchy surface for the exterior. We decompose the space of initial data, $\ms V^\infty$, on $\Sigma$ into \(L^2\)-orthogonal parts \(\ms V^\infty = \ms V^\infty_{\rm odd} \oplus \ms V^\infty_{\rm even} \) as follows. If the background is static, then $i$ is purely a ``$t$-reflection,'' and we define \(P \defn \lb( p_{ab}, 0 \rb) \in \ms V^\infty_{\rm odd} \) and \(Q \defn \lb( 0, q_{ab} \rb) \in \ms V^\infty_{\rm even} \). In the stationary and axisymmetric case, we first decompose \(p_{ab}\) and \(q_{ab}\) into their ``axial" and ``polar" parts with respect to the axial Killing fields \(\phi^a_\Lambda\) as follows\footnote{Note that this is a purely local decomposition into the parts that are parallel and orthogonal to the axial Killing fields.}:
\be\begin{split}
	p_{ab} & = 2\lambda_{(a}^\Lambda\phi_{b)\Lambda} + \beta_{ab} + \gamma^{\Lambda\Theta} \phi_{a\Lambda}\phi_{b\Theta} \\
	q_{ab} & = 2\alpha_{(a}^\Lambda\phi_{b)\Lambda} + \mu_{ab} + \nu^{\Lambda\Theta} \phi_{a\Lambda}\phi_{b\Theta}
\end{split}\ee
with \(\alpha_a^\Lambda\phi^a_\Theta = 0 = \lambda_a^\Lambda\phi^a_\Theta\); $\beta_{ab} = \beta_{(ab)}$, $\mu_{ab} = \mu_{(ab)}$; $\beta_{ab}\phi^a_\Lambda = 0 = \mu_{ab}\phi^a_\Lambda$; and \(\gamma^{\Lambda\Theta} = \gamma^{(\Lambda\Theta)}\), \(\nu^{\Lambda\Theta} = \nu^{(\Lambda\Theta)}\). Then the ($t$-$\phi$)-reflection odd and even parts of the perturbation respectively are:
\be\label{eq:tphi-decomp}\begin{split}
	P & \defn \lb( \beta_{ab} + \gamma^{\Lambda\Theta} \phi_{a\Lambda}\phi_{b\Theta},~ - 2\alpha_{(a}^\Lambda\phi_{b)\Lambda} \rb) \in \ms V^\infty_{\rm odd} \\
	Q & \defn \lb( 2\lambda_{(a}^\Lambda\phi_{b)\Lambda},~ \mu_{ab} + \nu^{\Lambda\Theta} \phi_{a\Lambda}\phi_{b\Theta} \rb) \in \ms V^\infty_{\rm even}
\end{split}\ee
The negative sign in front of the \(\alpha\) in the definition of  \(P\) was chosen so that the transformation $(p,q) \mapsto (P,Q)$ is canonical with respect to \(\mc S\) defined in \autoref{eq:symplectic-op-defn}. 

Since the linearized constraint equations as well as our boundary conditions and gauge conditions are invariant under $i$, they cannot couple $P$ and $Q$, so $P$ and $Q$ may be viewed as independent perturbations, i.e., if $(P,Q)$ is a perturbation satisfying the constraints, boundary conditions, and gauge conditions, then $(P,0)$ and $(0,Q)$ also satisfy the constraints, boundary conditions, and gauge conditions. 

The {\em canonical energy} is the map $\ms E : \ms V^\infty \times \ms V^\infty \to \mathbb{R}$ defined by
\be
\ms E (\widetilde{\delta X},\delta X) = \Omega_\Sigma \lb(\widetilde{\delta X},\Lie_t \delta X \rb)
\label{canendef}
\ee
Although the definition of $\ms E$ is not manifestly symmetric in $\widetilde{\delta X}$ and $\delta X$, it is, in fact, easily seen to be symmetric (see Prop.2 of \cite{HW-stab}). An explicit formula for $\ms E$ can be found in Eq.86 of \cite{HW-stab}. Now, $\ms E$ is constructed from the background spacetime, so it is invariant under the reflection isometry, $i$, in the sense that for any perturbations $\widetilde{\delta X}$ and $\delta X$ in \(\ms V^\infty\), we have ${\ms E} (i^* \widetilde{\delta X}, i^* \delta X) = {\ms E} (\widetilde{\delta X}, \delta X)$. It follows that under the decomposition \(\ms V^\infty = \ms V^\infty_{\rm odd} \oplus \ms V^\infty_{\rm even}\), \(\ms E\) cannot contain any ($P$-$Q$)-cross-terms. Thus, $\ms E$ splits up into two quadratic forms $\ms K: \ms V^\infty_{\rm odd} \times \ms V^\infty_{\rm odd} \to {\mathbb R}$ and $\ms U: \ms V^\infty_{\rm even} \times \ms V^\infty_{\rm even} \to {\mathbb R}$ such that
\be
\ms E[(\tilde{P}, \tilde{Q}), (P,Q)] = {\ms K} (\tilde{P}, P) + {\ms U} (\tilde{Q}, Q) \, ,
\ee
where
\be
{\ms K} (\tilde{P}, P) = \ms E[(\tilde{P}, 0), (P,0)] 
\label{KE}
\ee
\be
{\ms U} (\tilde{Q}, Q) = \ms E[(0, \tilde{Q}), (0,Q)] \, .
\ee
We refer to $\ms K$ and $\ms U$, respectively, as the \emph{kinetic energy} and \emph{potential energy} of the perturbation.

In the simple case of a static background, where $P = (p_{ab}, 0)$, it follows immediately from Eq.86 of \cite{HW-stab} with \(q_{ab} = 0\) that
\be\label{eq:KE-static}
	\ms K = 2\int_\Sigma N \lb[ p_{ab}^2 - \tfrac{1}{d-1}p^2\rb]
\ee
This expression is not manifestly positive definite since the second term in the integrand may dominate the first. However, $p_{ab}$ is not arbitrary but is subject to the constraint
\be
D^b p_{ab} = 0
\ee
as well as the boundary condition $p_{rr} = 0$ (see \autoref{eq:bound-defn-expodd}).

In the general case of a stationary-axisymmetric black hole, the kinetic energy is obtained by evaluating Eq.86 of \cite{HW-stab} for a perturbation of the form
\be\label{Ppert}
p_{ab} = \beta_{ab} + \gamma^{\Lambda\Theta} \phi_{a\Lambda}\phi_{b\Theta} \eqsp q_{ab} =  2\alpha_{(a}^\Lambda\phi_{b)\Lambda}
\ee
Using the background constraints and integrating by parts using \autoref{eq:bound-defn}, we can write the kinetic energy in the form
\be\begin{split}
\label{Kgen}
	\ms K & =  \int_\Sigma N \lb[\frac{1}{2}(D_cq_{ab})^2 -  D_cq_{ab}D^aq^{cb} \rb] + 2\int_\Sigma N \lb[ p_{ab}^2  -\frac{1}{d-1} \lb( p + \tfrac{1}{\sqh}~ q_{ab}\pi^{ab} \rb)^2 \rb.   \\
		& \quad \lb. + 4 \lb( p_{ab} + \tfrac{1}{\sqh}~q_{c(a}\pi^c_{b)} \rb)\tfrac{1}{\sqh}~ q^{ad}\pi_d^b  + \tfrac{1}{h}~ q_{cd}\pi^c_a\pi^d_b q^{ab} \rb] + 2\int_\Sigma \lb( p^{ab} + \tfrac{1}{\sqh}~2q^{ad}\pi^b_d \rb) \Lie_N q_{ab}  \, .
\end{split}\ee
The linearized constraints take the form
\be
	\mf{c}_a = D^bp_{ab} - \pi^{b\Lambda}\lb( 2D_{[a}\alpha_{b]\Lambda}  + \alpha_a^\Theta D_b\Phi_{\Lambda\Theta} \rb) = 0 \, ,
\ee
where \autoref{eq:Dphi} was used and, again, the form \autoref{Ppert} for $p_{ab}$ is understood. The linearized Hamiltonian constraint \(\mf c = 0\) is identically satisfied for perturbations of the form \autoref{Ppert}.

We now rewrite \autoref{Kgen} in a more explicit and useful form as follows. With the substitution \autoref{Ppert}, the first term in the integrand of \autoref{Kgen} becomes
\be\begin{split}\label{K1}
	{\ms K}_1 &\equiv \int_\Sigma N \lb[\frac{1}{2}(D_cq_{ab})^2 -  D_cq_{ab}D^aq^{cb} \rb] \\
		&= 2\int_\Sigma N \Phi^{\Lambda\Theta}\lb[(D_{[a}\alpha_{b]\Lambda})(D^{[a}\alpha^{b]}_\Theta) - \frac{1}{2} (\alpha^{a\Xi}D_a\Phi_{\Lambda\Gamma})(\alpha^{b\Gamma}D_b\Phi_{\Theta\Xi}) + (\alpha^{a\Xi} D^b\Phi_{\Lambda\Xi} )(D_a\alpha_{b\Theta}) \rb]   \, .
\end{split}\ee
The last term of \autoref{K1} can be written as
\be
		2\int_\Sigma N \Phi^{\Lambda\Theta}(\alpha^{a\Xi} D^b\Phi_{\Lambda\Xi} )(D_a\alpha_{b\Theta}) = 2\int_\Sigma N\Phi^{\Lambda\Theta} \lb[ (\alpha^{a\Xi} D^b\Phi_{\Lambda\Xi} )(D_b\alpha_{a\Theta}) + 2(\alpha^{a\Xi} D^b\Phi_{\Lambda\Xi} )(D_{[a}\alpha_{b]\Theta}) \rb]
\label{K13}
\ee
This expression can be simplified by using the relation
\be
R_{ab}\phi^b_\Lambda = -\tfrac{1}{2} \phi_{a\Theta} D^b\lb(  \Phi^{\Theta\Xi} D_b\Phi_{\Lambda\Xi} \rb)
\ee
(which holds by virtue  of $\phi^a_\Lambda$ being Killing fields) and then eliminating $R_{ab}$ using the background ADM equation \autoref{eq:Ric-id} to obtain
\be
		D^a\lb( N\Phi^{\Lambda\Xi}D_a\Phi_{\Theta\Xi} \rb) = -4N\pi_a^\Lambda\pi^a_\Theta \, .
\ee
Using this relation, we simplify \autoref{K13} as follows:
\be\begin{split}
		& 2\int_\Sigma N\Phi^{\Lambda\Theta} (\alpha^{a\Xi} D^b\Phi_{\Lambda\Xi} )(D_b\alpha_{a\Theta}) = -2 \int_\Sigma \lb[ D^b\lb( N\Phi^{\Lambda\Theta}D_b \Phi_{\Lambda\Xi} \rb)\alpha^{a\Xi}\alpha_{a\Theta} + N \alpha_a^\Lambda D^b\alpha^{a\Xi}  D_b\Phi_{\Lambda\Xi} \rb]\\
		& = 8\int_\Sigma N \lb(\alpha^{a\Lambda}\alpha_a^\Theta \rb)\lb(\pi_{a\Lambda}\pi^a_\Theta \rb) - 2\int_\Sigma N\Phi^{\Lambda\Theta} (\alpha^{a\Xi} D^b\Phi_{\Lambda\Xi} )(D_b\alpha_{a\Theta}) -2\int_\Sigma \alpha^{a\Lambda}\alpha_{a\Theta}D_b \Phi^{\Theta\Xi}D^b\Phi_{\Lambda\Xi}  \\
		& = 4\int_\Sigma N \lb(\alpha^{a\Lambda}\alpha_a^\Theta \rb)\lb(\pi_{a\Lambda}\pi^a_\Theta \rb) - \int_\Sigma \alpha^{a\Lambda}\alpha_{a\Theta}D_b \Phi^{\Theta\Xi}D^b\Phi_{\Lambda\Xi}
	\end{split}\ee
Thus, we obtain
	\[\begin{split}
 \ms K_1 & = 2\int_\Sigma N \Phi^{\Lambda\Theta}\lb[(D_{[a}\alpha_{b]\Lambda} + \alpha_{[a}^\Xi D_{b]}\Phi_{\Lambda\Xi})(D^{[a}\alpha^{b]}_\Theta + \alpha^{\Gamma[a}D^{b]}\Phi_{\Theta\Gamma}) \rb] + 4\int_\Sigma N \lb(\alpha^{a\Lambda}\alpha_a^\Theta \rb)\lb(\pi_{a\Lambda}\pi^a_\Theta \rb) \\
			& = 2\int_\Sigma N \Phi_{\Lambda\Theta}(D_{[a}\alpha_{b]}^\Lambda)(D^{[a}\alpha^{b]\Theta}) + 4\int_\Sigma N \lb(\alpha^{a\Lambda}\alpha_a^\Theta \rb)\lb(\pi_{a\Lambda}\pi^a_\Theta \rb) 
	\end{split}\]
To simplify the remaining terms in \autoref{Kgen} we use
\be
\sqh q_{ca}\pi^c_b = \lb( \alpha_a^\Lambda\pi_{b\Lambda} + \alpha_c^\Lambda\pi^{c\Theta} \phi_{a\Lambda}\phi_{b\Theta} \rb) \, .
\ee
In addition,
\be
 \Lie_N q_{ab} = \bar N^\Lambda\Lie_{\phi_\Lambda} q_{ab} + 2q_{c(a}\phi^c_\Lambda D_{b)}\bar N^\Lambda = - 4N\alpha_{(a}^\Lambda\pi_{b)\Lambda}
\ee
where \(\Lie_{\phi_\Lambda}q_{ab} = 0 \) by axisymmetry, and \autoref{eq:shift-id} was used. We obtain
\be\begin{split}
\ms K_2 \equiv &~ 2\int_\Sigma N \lb[ p_{ab}^2  -\frac{1}{d-1} \lb( p + \tfrac{1}{\sqh}~ q_{ab}\pi^{ab} \rb)^2 \rb.  \\
		&\quad \lb. + 4 \lb( p_{ab} + \tfrac{1}{\sqh}~q_{c(a}\pi^c_{b)} \rb)\tfrac{1}{\sqh}~ q^{ad}\pi_d^b  + \tfrac{1}{h}~ q_{cd}\pi^c_a\pi^d_b q^{ab} \rb] + 2\int_\Sigma \lb( p^{ab} + \tfrac{1}{\sqh}~2q^{ad}\pi^b_d \rb) \Lie_N q_{ab}  \\
		= &~ 2\int_\Sigma N \lb[ \beta_{ab}^2 + \lb( \gamma^{\Lambda\Theta} + 2\alpha_a^{(\Lambda}\pi^{\Theta) a} \rb)\lb( \gamma_{\Lambda\Theta} + 2\alpha_{a(\Lambda}\pi^a_{\Theta)} \rb)   -\tfrac{1}{d-1}\lb( \beta + \gamma^\Lambda_\Lambda + 2\alpha_a^\Lambda\pi^a_\Lambda \rb)^2 \rb] \\
	&\quad - 4\int_\Sigma N \lb(\alpha^{a\Lambda}\alpha_a^\Theta \rb)\lb(\pi_{a\Lambda}\pi^a_\Theta \rb)
\end{split}\ee
Putting these results together, we obtain our final expression for $\ms K$
\be\begin{split}\label{eq:KE-station}
		\ms K & = 2\int_\Sigma N \lb[ \Phi_{\Lambda\Theta}(D_{[a}\alpha_{b]}^\Lambda)(D^{[a}\alpha^{b]\Theta}) + \beta_{ab}^2 + \lb( \gamma^{\Lambda\Theta} + 2\alpha_a^{(\Lambda}\pi^{\Theta) a} \rb)\lb( \gamma_{\Lambda\Theta} + 2\alpha_{a(\Lambda}\pi^a_{\Theta)} \rb) \rb. \\
		&\qquad\qquad \lb. -\frac{1}{d-1}\lb( \beta + \gamma^\Lambda_\Lambda + 2\alpha_a^\Lambda\pi^a_\Lambda \rb)^2 \rb] 
	\end{split}\ee

We are now in a position to state and prove our first theorem:

\begin{thm} [Positivity of kinetic energy] \label{thm:positive-KE}
For axisymmetric perturbations of a stationary-axisymmetric black hole background, the kinetic energy \(\ms K\), \autoref{eq:KE-station}, is a positive-definite symmetric bilinear form on \(\ms V^\infty_{\rm odd}\).
\end{thm}
\begin{proof} It is instructive to first give the proof for the case of a static background, where $i$ is given by a $t$-reflection and the kinetic energy is given by the much simpler expression \autoref{eq:KE-static}. In that case, let \(\xi\) be the solution to the following boundary value problem (see \cite{Chr-De} and \cite{Can-Die}):
\be\label{eq:xi-defn}
\triangle \xi = \frac{1}{d-1}p \eqsp \xi\vert_B = 0 \eqsp \xi \sim O(1/\rho^{d-3})\vert_\infty \, \, .
\ee
Define \( \hat P = \lb(\hat p_{ab}, 0\rb) \in \ms P_{\rm gr}^\infty \) by 
\be
 \hat P \defn P - \mc S^*\mc L^* \lb(\xi, 0 \rb)
\ee
where \(\mc L^*\) was defined by \autoref{eq:L*-defn}, i.e., we define
\be
\hat{p}_{ab} \defn p_{ab}  -  D_aD_b\xi + h_{ab}\triangle\xi + R_{ab} \xi
\label{phat}
\ee
where we have used the fact that for the static background, we have $\pi^{ab} = 0$ and $R=0$. Note that $\hat p_{ab}$ satisfies the constraint $D^a \hat{p}_{ab} = 0$ and the boundary condition $\hat p_{rr}|_B = 0$ (see \autoref{eq:bound-defn-expodd}). Note also that $\hat p = \hat p^a{}_a = 0$. We use \autoref{phat} to eliminate $p_{ab}$ in favor of $\hat p_{ab}$ and $\xi$ in \autoref{eq:KE-static}. Integrating by parts and using the boundary conditions on \(\xi\) (from \autoref{eq:xi-defn}) and \(\hat p_{rr}\vert_B = 0\), we find that the $\xi$-$\xi$ terms and the $\xi$-$\hat p$ cross-terms vanish. (This is essentially the same calculation as done in \autoref{eq:bound-integral} and is an expression of the gauge invariance of $\ms K$ with respect to gauge transformations that respect the boundary conditions 
\autoref{eq:bound-defn}.) Thus, we obtain
\be\label{eq:K-static-pos}
\ms K = 2\int_\Sigma N \hat p^{ab} \hat p_{ab} \geq 0 
\ee
Thus, $\ms K$ is non-negative and vanishes if and only if \(\hat p_{ab}= 0 \). However, if $\hat p_{ab} = 0$, then \( P = \mc S^*\mc L^* \lb(\xi, 0 \rb) \in \ms W_g \). But \(P \in \ms V^\infty_{\rm odd} \subset \ms V \subset \ms W_g^\perp\). Thus, \(\ms K = 0\) if and only if \(P = 0\).

The proof proceeds in an exactly parallel manner for a general stationary-axisymmetric background, where $\ms K$ is given by \autoref{eq:KE-station}. We now let \(\xi\) be the solution to the boundary value problem
\be
\lb(\triangle - 2\pi_a^\Lambda\pi^a_\Lambda\rb)\xi = \tfrac{1}{d-1}\lb( \beta + \gamma^\Lambda_\Lambda + 2\alpha_a^\Lambda\pi^a_\Lambda \rb) \eqsp \xi\vert_B = 0 \eqsp \xi \sim O(1/\rho^{d-3})\vert_\infty
\ee
We again define \(\hat P  \in \ms P_{\rm gr}^\infty\) by 
\be
 \hat P \defn P - \mc S^*\mc L^* \lb(\xi,~ 0 \rb) 
\ee
and we decompose $\hat P$ as
\be
\hat P = \lb( \hat\beta_{ab} + \hat\gamma^{\Lambda\Theta} \phi_{a\Lambda}\phi_{b\Theta} ,~ - 2\hat\alpha_{(a}^\Lambda\phi_{b)\Lambda} \rb) \, .
\ee
By a direct computation using the above definitions, we find \( \hat\beta + \hat\gamma^\Lambda_\Lambda + 2\hat\alpha_a^\Lambda\pi^a_\Lambda = 0 \). In parallel with the arguments in the static case we obtain \( \ms K(P, P) = \ms K(\hat P, \hat P)\). Hence, we obtain
\be\label{eq:K-station-pos}
		\ms K = 2\int_\Sigma N \lb[ \Phi_{\Lambda\Theta}(D_{[a}\hat\alpha_{b]}^\Lambda)(D^{[a}\hat\alpha^{b]\Theta}) + \hat\beta_{ab}^2 + \lb( \hat\gamma^{\Lambda\Theta} + 2\hat\alpha_a^{(\Lambda}\pi^{\Theta) a} \rb)\lb( \hat\gamma_{\Lambda\Theta} + 2\hat\alpha_{a(\Lambda}\pi^a_{\Theta)} \rb) \rb] \geq 0
\ee
which shows that \(\ms K\) is non-negative. Furthermore, $\ms K$ vanishes if and only if \( \hat\alpha_a^\Lambda = D_a\bar\xi^\Lambda \), \(\hat\beta_{ab} = 0\) and \( \hat\gamma^{\Lambda\Theta} = -2D_a\bar\xi^{(\Lambda}\pi^{\Theta)a}  \). This means \(\ms K\) vanishes iff \( \hat P = \mc S^*\mc L^*\lb(0, -\bar\xi^\Lambda\phi^a_\Lambda \rb)\) i.e. \( P = \mc S^*\mc L^* \lb(\xi, -\bar\xi^\Lambda\phi^a_\Lambda \rb) \in \ms W_g \). But \(P \in \ms V^\infty_{\rm odd} \subset \ms V \subset \ms W_g^\perp\). So, \(\ms K = 0\) if and only if \(P = 0\), i.e., the kinetic energy \(\ms K\) is a positive-definite, symmetric, quadratic form on \(\ms V^\infty_{\rm odd}\), as we desired to show.

\end{proof}

\noindent
{\bf Remarks:} 

\begin{enumerate}

\item As shown in \cite{SW-tphi}, a reflection isometry, $i$, exists for any subgroup of stationary-axisymmetric isometries that acts trivially and is such that its generators include the stationary Killing field and the horizon Killing field. If more than one choice of $i$ exists, then \autoref{thm:positive-KE} implies correspondingly stronger results. In particular, if one has a static, axisymmetric black hole, one has both a $t$-reflection and a ($t$-$\phi$)-reflection isometry. If we define $\ms K$ and $\ms U$ with respect to the $t$-reflection isometry, then theorem 1 tells us that $\ms K$ is positive definite. However, if we apply theorem 1 to the ($t$-$\phi$)-reflection isometry, we learn that, in addition, the potential energy, $\ms U$, is also positive definite on ``axial'' metric perturbations, i.e., metric perturbations of the form $q_{ab} = 2\alpha_{(a}^\Lambda\phi_{b)\Lambda}$.

\item The transformation $P \to \hat P = P - \mc S^*\mc L^* \lb(\xi,~ 0 \rb)$ is just a gauge transformation corresponding to making a normal displacement of the hypersurface $\Sigma$ by $\xi$. The condition \( \hat\beta + \hat\gamma^\Lambda_\Lambda + 2\hat\alpha_a^\Lambda\pi^a_\Lambda = 0 \) is just the condition that $\hat{\delta} \pi =  \sqh \hat p + \hat q_{ab}\pi^{ab} = 0$. Thus, writing the expression for $\ms K$ in terms of $\hat P$ corresponds to working in a gauge\footnote{Note, however, that this gauge choice is not compatible with the gauge conditions that we imposed in \autoref{eq:gauge-conds}.} where $\Sigma$ is a maximal slice in the perturbed spacetime. The fact that $\ms K(P, P) = \ms K(\hat P, \hat P)$ is a manifestation of the gauge invariance of $\ms K$. 

\item {\bf Local Penrose Inequality}: As shown in \cite{HW-stab}, positivity of the canonical energy $\ms E$ on $\ms V^\infty$ is equivalent to the satisfaction of a local Penrose inequality. Consider perturbations of a Schwarzschild black hole with $d \leq 7$, with $i$ chosen to be the $t$-reversal isometry. The fact that the Riemannian Penrose inequality holds \cite{Bray-Lee} implies that $\ms U$ is positive. However, since there do not exist nontrivial stationary perturbations of Schwarzschild with $\delta M = \delta J = 0$ \cite{Carter}, \cite{IK-stab}, it follows that $\ms U$ is non-degenerate \cite{HW-stab}. Thus, $\ms U$ is positive definite.
Since we have just shown that $\ms K$ is positive definite, it follows that $\ms E$ is positive definite. This implies that the Schwarzschild black hole is a strict local minimum of mass at fixed area of the apparent horizon.

\end{enumerate}


\section{Time Evolution and the Inverse Kinetic Energy Hilbert Space $\ms H$}\label{sec:evolution}

In this section, we will introduce the time evolution operators, $\mc K$ and $\mc U$, which will be seen to correspond to the quadratic forms, $\ms K$ and $\ms U$, defined in the previous section. We will then use the positive definiteness of $\ms K$ to define a new Hilbert space $\ms H$ that we will use in our stability analysis. The Hilbert space $\ms H$ will {\em not} contain all of $\ms V^\infty_{\rm even}$, but it will contain all solutions of the form $\Lie_t \delta X$ with $\delta X \in \ms V^\infty_{\rm odd}$.

The time evolution of the background with respect to \(t^\mu\) in \autoref{eq:evol-bg} is a gauge transformation generated by the lapse and the shift \( \svp N \equiv \lb( N, N^a \rb)\), i.e. \autoref{eq:evol-bg} takes the form
\be
\dot X = \mc S^*\mc L^*\svp N \, . 
\ee
The evolution equations for smooth perturbations are obtained by linearizing this relation, i.e., 
\be
 \dot{\delta X} = \mc S^* \mc L^* \svp n + \mc S^* \delta (\mc L^*) \svp N 
\ee 
where 
\be
\svp n \defn \lb( \delta N, \delta N^a \rb)
\ee
and \(\delta(\mc L^*)\) denotes the linearization of the operator \(\mc L^*\) given by \autoref{eq:L*-defn}. By a direct computation of $\delta (\mc L^*)$ from \autoref{eq:L*-defn}, we obtain
\be\label{eq:evol-q}\begin{split}
	\dot q_{ab} & = \hat\delta_nh_{ab} + 2N\lb( p_{ab}-\tfrac{1}{d-1}p~ h_{ab} \rb) - \tfrac{1}{\sqh} Nq \lb( \pi_{ab} - \tfrac{1}{d-1}\pi h_{ab} \rb) \\
				&\quad + \tfrac{2}{\sqh}N\lb( 2 q_{c(a}\pi^c_{b)} - \tfrac{1}{d-1}q_{cd}\pi^{cd} h_{ab} - \tfrac{1}{d-1}\pi q_{ab} \rb) + \Lie_N q_{ab}
\end{split}\ee
\be\label{eq:evol-p}\begin{split}
	\dot p_{ab} & = (\tfrac{1}{\sqh}\hat\delta_n\pi^{cd})h_{ac}h_{bd} +\tfrac{1}{2} N \lb( \triangle q_{ab} + D_aD_bq \rb) - \tfrac{1}{2}D^cN \lb( D_aq_{bc} + D_bq_{ac} - D_cq_{ab} + h_{ab} D_cq \rb)\\
		& \quad\, - N\lb( R_{c(ab)d}q^{cd} + {R^c}_{(a}q_{b)c} -h_{ab}R^{cd}q_{cd} \rb) + h_{ab} D^cND^dq_{cd} - ND_{(a}D^cq_{b)c} \\
		&\quad - \lb( NR - \triangle N \rb)q_{ab} + \tfrac{3}{2} Nh_{ab}\lb( -R^{cd}q_{cd} - \triangle q + D^cD^dq_{cd} \rb)	\\
		& \quad + \tfrac{2}{h}Nq\lb( \pi_{ac}\pi^c_b - \tfrac{1}{d-1}\pi\pi_{ab} \rb) - \tfrac{2}{h}N \lb[ q_{cd}\pi^c_a\pi^d_b + 2q_{c(a}\pi_{b)d}\pi^{cd} -\tfrac{1}{d-1}\lb( q_{cd}\pi^{cd}\pi_{ab} + 2 q_{c(a}\pi^c_{b)}\pi \rb)\rb] \\
		& \quad -\tfrac{2}{\sqh}N\lb[ 2p_{c(a}\pi^c_{b)} - \tfrac{1}{d-1}\lb( p\pi_{ab} + \pi p_{ab} \rb) \rb] -\tfrac{1}{2\sqh}q \lb[ D_c\lb( N^c\pi_{ab} \rb) - 2\pi_{c(a}D^cN_{b)}  \rb] \\
		& \quad + D_c\lb[ N^c \lb(p_{ab} + \tfrac{2}{\sqh} q_{d(a}\pi^d_{b)}\rb) \rb] - 2 \lb[ p_{c(a} + \tfrac{1}{\sqh}q_{dc}\pi^d_{(a} + \tfrac{1}{\sqh}\pi^d_c q_{d(a} \rb] D^cN_{b)} - \tfrac{2}{\sqh}\Lie_N q_{d(a}\pi_{b)}^d \\
		& \quad + \tfrac{2}{\sqh}q^c_d \lb[ D_{(c}N_{a)}\pi^d_b + D_{(c}N_{b)}\pi^d_a  \rb]
\end{split}\ee
where
\be\begin{split}
	\begin{pmatrix} \tfrac{1}{\sqh}\hat\delta_n \pi^{ab} \\ \hat\delta_nh_{ab} \end{pmatrix} \defn 
	\mc S^*\mc L^*\begin{pmatrix} n \\ n^a \end{pmatrix}
\end{split}\ee
is just the gauge transformation generated by \(\svp n\). In the above equations, we have used the background constraints and evolution equations, as well as the stationarity (but not axisymmetry) of the background.

The perturbed lapse and shift $\svp n$ are to be chosen so that the solution lies in $\ms V^\infty$. The linearized constraints \autoref{linconst} are preserved under time evolution, and thus \(\mc L\dot{\delta X} = 0\) holds as an identity. Similarly, \(\dot{\delta X}\) identically satisfies the boundary conditions \autoref{eq:bound-defn} as can be checked by an explicit computation. Thus, \(\dot{\delta X} \in \ms V^\infty_c\) and we only need to impose the conditions \autoref{eq:gauge-conds} so that \(\dot{\delta X} \in \ms V^\infty_g\). This requires that $\svp n$ satisfies the equation
\be
\label{eq:n}
\mc L\mc L^*\svp n = - \mc L \delta(\mc L^*)\svp N
\ee
This equation takes the form
\begin{subequations}\label{eq:LL*}\begin{align}
	(d-1) \triangle^2 n  & = \ldots \label{eq:LL*-n} \\
	2 D^bD_{(a}n_{b)}  & = \ldots \label{eq:LL*-na}
\end{align}\end{subequations}
where the dots denote terms of lower derivative order in \(\svp n\) and source terms depending on \(\delta X\). 
Thus, \autoref{eq:n} is of elliptic character. The boundary conditions needed for 
the right hand side of \autoref{eq:evol-q}-\autoref{eq:evol-p} to lie in $\ms V^\infty$ are that \(\mc S^* \mc L^* \svp n \in \ms W_g\) and that, in addition, $\svp n$ satisfy the boundary conditions arising from \autoref{eq:gauge-conds-adm} and \autoref{eq:gauge-conds-B}. The conditions arising from \autoref{eq:gauge-conds-B} take the form
\begin{subequations}\label{eq:bound-conds-LL*}\begin{align}
	0 = \delta\varepsilon(\mc S \dot{\delta X})\vert_B & = \tfrac{1}{2}(d-1) d_r^2 n + \tfrac{1}{2}\lb[(d-2)\ms D^2 + \tfrac{1}{2} \ms R  \rb]n + \ldots + \delta\varepsilon(\delta(\mc L^*)\svp N) \label{eq:bound-conds-LL*-area} \\
	0 = \delta\varpi_{AB}(\mc S \dot{\delta X})\vert_B & = \ms D_{[A} d_r n_{B]} + \ldots + \delta\varpi_{AB}(\delta(\mc L^*)\svp N) \label{eq:bound-conds-LL*-varpi}
\end{align}\end{subequations}
where we have again suppressed the lower derivative terms in \(\svp n\) and kept the source terms depending on \(\delta X\) in symbolic form.

The problem of solving the above system for $\svp n$ corresponds precisely to applying the projection map\footnote{It can be verified that $\mc S^* \delta (\mc L^*) \svp N$ lies in $\ms V_c^\infty$ (see the proof of \autoref{prop:evol} below), so no $\Pi_c$ projection is needed.} $\Pi_g$ to $\mc S^* \delta (\mc L^*) \svp N$. It follows from the results of Appendix \ref{sec:proj-ops} that a unique solution for \(\svp n \) exists satisfying \autoref{eq:n} and the above boundary conditions. Consequently, the evolution equations \autoref{eq:evol-q} and \autoref{eq:evol-p} take the form 
\be
\dot{\delta X} = \mc E (\delta X) \, .
\label{dotdX}
\ee
where $\mc E : \ms V^\infty \to \ms V^\infty$ is a linear map. Note, however, that since we have to solve the elliptic system \autoref{eq:n} for $\svp n$ to determine the action of $\mc E$ on $\delta X$, $\mc E$ is not a local differential operator. 

We have the following proposition:

\begin{prop}\label{prop:evol}
Given any $\delta X_0 \in \ms V^\infty$ there exists a unique solution $\delta X(t) \in \ms V^\infty$ to \autoref{dotdX} that varies smoothly with $t$ in the topology of $\ms V^\infty$ and is such that $\delta X(0) = \delta X_0$.
\end{prop}
\begin{proof}
Uniqueness is obvious because the difference between two such solutions yields a solution to the linearized Einstein equation with vanishing initial data. Such a solution must be pure gauge, but the gauge has been completely fixed in $\ms V^\infty$, so the difference between two such solutions must vanish. To show existence, we solve the linearized Einstein equation in linearized harmonic gauge to obtain a metric perturbation $\delta g_{\mu \nu}$ with initial data $\delta X_0$. This solution will induce initial data $\widetilde{\delta X}(t)$ on each slice $\Sigma_t$. On account of the wave equation character of the linearized Einstein equation in this gauge, standard energy estimates (see, e.g., Lemma 7.4.6 of \cite{Hawking-Ellis}) can be used to show that $\widetilde{\delta X}(t)$ will lie in $\ms P_{\rm gr}^\infty$ and will vary smoothly with $t$. However, $\widetilde{\delta X}(t)$ must also satisfy the linearized constraints \autoref{linconst} and the conditions \autoref{admcond} because $\delta X_0$ satisfies these conditions and they are automatically preserved under time evolution. In addition ``time evolution'' (i.e., the evaluation of the solution on $\Sigma_t$ using the identification of $\Sigma_t$ with $\Sigma$ given by orbits of $t^\mu$) maps $B$ into $B$ (see \autoref{fig:spacetime}) and merely rotates it via the action of axial Killing fields. Since the perturbation is axisymmetric, it follows that \autoref{horcond} holds at all times if it holds initially. Thus, $\widetilde{\delta X}(t) \in \ms V_c$ for all $t$. The desired solution then is $\delta X(t) = \Pi_g \widetilde{\delta X}(t)$. Smoothness of $\delta X (t)$ in $t$ follows from the smoothness of 
$\widetilde{\delta X}(t)$ in $t$ together with the continuity of the gauge projection map $\Pi_g$ proven in Appendix \ref{sec:proj-ops}.
\end{proof}

The canonical energy, $\ms E$, is related to the evolution operator $\mc E$ by
\be
\ms E (\widetilde{\delta X}, \delta X) = \Omega_\Sigma \lb(\widetilde{\delta X},\mc E(\delta X) \rb) =\inp{\widetilde{\delta X}, \mc S\mc E (\delta X)}
\label{canen2}
\ee
for all $\widetilde{\delta X}, \delta X \in \ms V^\infty$. This equation follows immediately from the definition of $\ms E$, \autoref{canendef}, together with \autoref{dotdX} and \autoref{eq:inner-prod-defn}-\autoref{eq:symplectic-op-defn}. The relation \autoref{canen2} between $\ms E$ and the time evolution operator $\mc E$ expresses the fact that $\ms E$ is a Hamiltonian for the linearized equations. 

We now make use of the reflection isometry $i$ of the background spacetime. As in the previous section, we decompose initial data $(p_{ab}, q_{ab}) \in \ms V^\infty$ into its odd part, $P \in \ms V^\infty_{\rm odd}$, and even part, $Q \in \ms V^\infty_{\rm even} $ under the action of $i$ (see \autoref{eq:tphi-decomp}). Since the time evolution operator $\mc E$ is invariant under $i$, the gauge-fixed ADM evolution equations \autoref{eq:evol-q}-\autoref{eq:evol-p} take the form:
\begin{subequations}\label{eq:evol-odd-even}\begin{align}
	\dot Q = \mc K P  \label{eq:evol-K} \\
	\dot P = -\mc U Q \label{eq:evol-U}
\end{align}\end{subequations}
Here, the maps $\mc K$ and $\mc U$ act as
\be\begin{split}
	\mc K & : \ms V^\infty_{\rm odd} \to \ms V^\infty_{\rm even} \\
	\mc U & : \ms V^\infty_{\rm even} \to \ms V^\infty_{\rm odd} \, . \\
\end{split}\ee
We refer to $\mc K$ and $\mc U$, respectively, as the \emph{kinetic and potential time evolution operators}.  As in the case of $\mc E$, the operators \(\mc K\) and \(\mc U\) are not local differential operators since the gauge fixing procedure involves solving an elliptic equation. 

We now return to \autoref{canen2} and decompose the perturbations $\widetilde{\delta X}$ and $\delta X$ into their odd and even parts under $i$. Using the fact that no ``cross-terms'' can arise on account of the fact that $i$ is an isometry of the background spacetime, we see that the kinetic, potential and canonical energy quadratic forms, $\ms K$ and $\ms U$, are given in terms of the kinetic and potential time evolution operators, $\mc K$ and $\mc U$, by
\be\begin{split}
	\ms K(\tilde P, P) &= \inp{\tilde P, \mc S\mc K P} \\
	\ms U(\tilde Q, Q) &= -\inp{\tilde Q, \mc S \mc U Q} \, .\\
\label{KandU}
\end{split}\ee
Note that it follows immediately from the symmetry of the quadratic forms $\ms K$ and $\ms U$ that both \(\mc S\mc K: \ms V^\infty_{\rm odd} \to \ms V^\infty_{\rm odd}\) and \(\mc S\mc U: \ms V^\infty_{\rm even} \to \ms V^\infty_{\rm even}\) are symmetric linear maps in the \( L^2 \)-inner product. 

Taking the time derivative of \autoref{eq:evol-K} and using \autoref{eq:evol-U}, we obtain
\be\label{ddotQ}
\ddot Q = -\mc A Q 
\ee
where
\be
\mc A \defn \mc K\mc U  : \ms V^\infty_{\rm even} \to \ms V^\infty_{\rm even}  \, .
\label{Adef}
\ee
We wish to use spectral methods to solve this equation, but in order to do so, we need to define an inner product that makes $\mc A$ a symmetric operator. This can be achieved as follows. Let $\mc K[\ms V_{\rm odd}^\infty] \subseteq \ms V^\infty_{\rm even}$ denote the range of the operator $\mc K$. By \autoref{thm:positive-KE}, \(\mc S\mc K\) is a positive-definite operator on \(\ms V^\infty_{\rm odd}\), so $\mc K$ has vanishing kernel. Thus for all \( Q \in \mc K[\ms V_{\rm odd}^\infty] \), there exists a unique \( P \in \ms V^\infty_{\rm odd} \) such that, \( Q = \mc K P \). Using this fact, define a new inner product, $\inp{ , }_{\ms H}$ on \( \mc K[\ms V_{\rm odd}^\infty] \) by:
\be\label{eq:H-inp-defn}
	\inp{\tilde Q , Q}_{\ms H} \defn \ms K(\tilde P, P) 
\ee
where $\tilde{P}$ and $P$ are such that \( \tilde Q = \mc K \tilde P \) and \( Q = \mc K P \). That this is indeed an inner product follows from the symmetry, bilinearity and positive-definiteness of \(\ms K\) (\autoref{thm:positive-KE}). Note that we can write this inner product in terms of the $L^2$ inner product
\autoref{eq:inner-prod-defn} and the operator $\mc K$ as
\be
\inp{\tilde Q , Q}_{\ms H} = \inp{\tilde P, \mc S\mc K P} = \inp{\tilde P, \mc SQ} = \inp{\mc S\tilde Q, P} 
\label{Hprod2}
\ee
Thus, formally, the new inner product \autoref{eq:H-inp-defn} corresponds to taking the matrix element of $(\mc S \mc K)^{-1}$ in the $L^2$ inner product \autoref{eq:inner-prod-defn}.

We now complete the space \(\mc K[\ms V_{\rm odd}^\infty] \) in the inner product $\inp{ , }_{\ms H}$ to obtain a Hilbert space \( \ms H \). Note that $\ms V_{\rm even}^\infty \not\subset \ms H$. However, $\ms H$ does contain all $Q \in \ms V_{\rm even}^\infty$ that are of the form $Q = \mc K P$ for $P \in \ms V_{\rm odd}^\infty$. In view of \autoref{eq:evol-odd-even}, this means that the even part of all perturbations that are of the form $\Lie_t\delta X'$ for some perturbation \(\delta X' \in \ms V^\infty\) will be represented in $\ms H$. As a consequence, our stability analysis of the next section will not apply to all perturbations in $\ms V^\infty$, but it will apply to all perturbations of the form $\Lie_t\delta X'$ with \(\delta X' \in \ms V^\infty\).

Finally, it will be convenient to complexify $\ms H$ in order to apply spectral methods. We continue denote the complexified Hilbert space as $\ms H$, i.e., we will not distinguish the complexification in our notation.


\section{Dynamics in $\ms H$; Negative Energy and Exponential Growth} \label{sec:exp}

The operator $\mc A$ of \autoref{Adef} is well defined as an operator $\mc A : \ms H \to \ms H$ with dense domain given by \( \dom\mc A = \mc K[\ms V_{\rm odd}^\infty] \). For $\tilde{Q}, Q \in \dom\mc A $, we have
\be
\inp{\tilde Q , \mc AQ}_{\ms H} =  \inp{\tilde Q ,  \mc K (\mc U Q)}_{\ms H} = \inp{\mc S\tilde Q, \mc U Q}  = -\inp{\tilde Q, \mc S\mc U Q} = \ms U(\tilde Q , Q)
\label{HU}
\ee
where \autoref{KandU} and the last equality of \autoref{Hprod2} were used. Thus \( \mc A \) is densely defined symmetric operator on \(\ms H\). Indeed, the entire purpose of introducing the inner product \autoref{eq:H-inp-defn} was to make $\mc A$ symmetric.

If \(\ms H\) were finite-dimensional, then the operator \(\mc A\) would be self-adjoint on \(\ms H\) and would admit an orthonormal basis of eigenvectors. This basis would diagonalize \autoref{ddotQ}, and we could immediately conclude that any eigenvector of \(\mc A\)  with a negative eigenvalue would correspond to an exponentially growing solution. However, \(\ms H\) is infinite-dimensional and \(\mc A\) is known only to be a symmetric operator defined on a dense domain in \(\ms H\), so the argument for exponential growth cannot be made so straightforwardly.

Clearly $\mc A : \ms H \to \ms H$ is a real operator, i.e., it commutes with the (anti-linear) complex conjugation operator on $\ms H$ that we have just introduced via our complexification at the end of the previous section.
We believe it likely that $\mc A$ satisfies the properties that it is bounded below and that it is essentially self-adjoint on the domain $\mc K[\ms V_{\rm odd}^\infty]$. If these properties hold, then $\mc A$ would have a unique self-adjoint extension $\tilde{\mc A}$ that also would be bounded below. Unfortunately, it does not appear to be straightforward to establish either of these properties. Fortunately, as we shall see, we can bypass the issue of essential self-adjointness of $\mc A$ by simply choosing an arbitrary self-adjoint extension $\tilde{\mc A}$; existence of such an extension is guaranteed by virtue of $\mc A$ being real. As we also shall see, we are able to bypass the issue of $\tilde{\mc A}$ being bounded below by using spectral projections.

Consider initial data of the form $(P_0=0, Q_0= \mc K P'_0)$ where $P'_0 \in \ms V_{\rm odd}^\infty$. Clearly, we have $Q_0 \in {\mc K}[\ms V^\infty_{\rm odd}]$, so $Q_0 \in \ms H$. Let $\delta X(t)$ denote the unique solution in $\ms V^\infty$ to \autoref{dotdX} with this initial data (see \autoref{prop:evol}). Let $Q_t \in \ms V_{\rm even}^\infty$ denote the ($t$-$\phi$)-even part of $\delta X(t)$. We claim first that $Q_t \in \ms H$ for all $t$. To show this, let $\delta X'$ be the unique solution in $\ms V^\infty$ with initial data $(P'_0, Q'_0 = 0)$. Then $\Lie_t \delta X'$ has initial data $(P_0=0, Q_0= \mc K P'_0)$, and hence $\delta X (t) = \Lie_t \delta X' (t)$ for all $t$. It follows that $Q_t$ takes the form $Q_t = \mc K P'_t$ where $P'_t \in \ms V_{\rm odd}^\infty$ is the ($t$-$\phi$)-odd part of $\delta X'(t)$. Thus, $Q_t \in \ms H$ for all $t$. Since $\dom \mc A = \ms V_{\rm even}^\infty \cap \ms H$, we also have $Q_t \in \dom \mc A$ for all $t$. Furthermore, when viewed as a one-parameter family of vectors in $\ms H$, $Q_t$ satisfies
\be
\frac{d^2Q_t}{dt^2} = - \mc A Q_t \, 
\label{QevolH}
\ee
(see \autoref{ddotQ}.)

Let \( \tilde{\mc A} \) be any self-adjoint extension of \( \mc A \) on \( \ms H \). Let \(\{E_\lambda\}\) denote the family of projection operators associated to $\tilde{\mc A}$ such that\footnote{The meaning of the integral (as a Lebesgue-Stieltjes integral) is explained, e.g., in Sec.120 of \cite{RN-book} and in \cite{Reed-Simon}. }
\be
	\tilde{\mc A} = \int\limits_{-\infty}^\infty\lambda dE_\lambda
	\label{specrep}
\ee
so that \(E_\lambda\) projects onto the spectral subspace of $\tilde{\mc A}$ in \((-\infty, \lambda]\). Define $Q_{t,\beta}$ by
\be
Q_{t,\beta} = (I - E_\beta) Q_t
\label{Qlambda}
\ee
so that $Q_{t,\beta}$ is the projection of $Q_t$ onto the spectral subspace $(\beta, \infty)$ of $\tilde{\mc A}$. Clearly, we have $Q_t = \lim\limits_{\beta \to -\infty} Q_{t,\beta}$ and $||Q_{t,\beta}||_{\ms H} \leq ||Q_t||_{\ms H}$ for all $\lambda$. Applying $(I - E_\beta)$ to \autoref{QevolH}, we obtain
\be
\frac{d^2Q_{t,\beta}}{dt^2} = - \tilde{\mc A} Q_{t,\beta} \, .
\label{QevolH2}
\ee

Our next result will require the following lemma:

\begin{lemma}\label{lem:unique-soln}
Let \(\tilde{\mc A} \) be an arbitrary self-adjoint operator on a Hilbert space $\ms H$. Let $R_t$ be a one parameter family of vectors in $\ms H$ such that (i) $R_t$ is twice differentiable in $t$, (ii) $R_0 = (d R_t/dt)|_0 = 0$ (iii) $R_t$ lies in $\dom \tilde{\mc A}$ for all $t$, and (iv) we have
\be
\frac{d^2R_t}{dt^2} = - \tilde{\mc A} R_t 
\ee
for all $t$. Then $R_t = 0$.
\end{lemma}
\begin{proof}
Let $\Lambda > 0$ and let $R_{t,\Lambda}$ be the projection of $R_t$ onto the spectral interval \([-\Lambda, \Lambda]\), i.e., 
\be
R_{t,\Lambda} = \lb(E_\Lambda - E_{-\Lambda}\rb) R_t
\label{sproj}
\ee
where $E_\lambda$ is the spectral family of projection operators associated with $\tilde{\mc A}$. Then $R_{t,\Lambda}$ satisfies all of the properties of $R_t$ stated in the hypothesis of this Lemma. Let
\be
f_\Lambda (t) = \norm{R_{t, \Lambda}}^2 + \norm{dR_{t, \Lambda}/dt}^2 \, .
\ee
Then we have:
\begin{eqnarray}
	\frac{df_\Lambda}{dt} & = 2 \Re \inp{dR_{t, \Lambda}/dt, R_{t, \Lambda}} + 2 \Re \inp{d^2R_{t, \Lambda}/dt^2, dR_{t, \Lambda}/dt} \\
&= 2 \Re \inp{dR_{t, \Lambda}/dt, R_{t, \Lambda}} + 2 \Re \inp{- \tilde{\mc A} R_{t, \Lambda}, dR_{t, \Lambda}/dt} 
\end{eqnarray}
where $\Re$ denotes the real part. Thus, we obtain
\be
\left| \frac{df_\Lambda}{dt} \right| \leq \norm{R_{t, \Lambda}}^2 + 2 \norm{dR_{t, \Lambda}/dt}^2 + \norm{\tilde{\mc A} R_{t, \Lambda}}^2  
\ee
However, $\norm{\tilde{\mc A} R_{t, \Lambda}}^2 \leq \Lambda^2 \norm{R_{t, \Lambda}}^2$ on account of the spectral projection \autoref{sproj}. Thus, we obtain
\be
\left| \frac{df_\Lambda}{dt} \right| \leq (\Lambda^2 + 2) f_\Lambda
\ee
Since \(f_\Lambda \geq 0 \) for all \(t \in \bb R\) and \(f_\Lambda\vert_{t=0} = 0\), this inequality implies that \(f_\Lambda(t) = 0 \) for all $t$, which implies that $R_{t,\Lambda} = 0$. However, the vanishing of $R_{t,\Lambda}$ for all $\Lambda$ implies that $R_t = 0$.
\end{proof}

Our instability results will be based upon the following proposition

\begin{prop}\label{prop:Qtformula} $Q_{t,\beta}$ as defined by \autoref{Qlambda} is given by
\be\label{eq:spectral-soln}
	Q_{t,\beta}  = \cos(t\tilde{\mc A}_+^{1/2})Q_{0, \beta}
				 + \Pi_0Q_{0,\beta}
			 + \cosh(t\tilde{\mc A}_-^{1/2}) Q_{0,\beta}
\ee
where $\tilde{\mc A}_+ = (I - E_0) \tilde{\mc A}$ is the positive part of $\tilde{\mc A}$, $\Pi_0$ is the projection onto the zero eigenvalue subspace\footnote{$\Pi_0$ will be nontrivial if and only if there exist stationary perturbations for which $\delta M = \delta J = 0$.} of $\tilde{\mc A}$, and $\tilde{\mc A}_- = -(E_0 - \Pi_0) \tilde{\mc A}$ is minus the negative part of $\tilde{\mc A}$ (so $\tilde{\mc A}_-$ is a positive operator).
\end{prop}
\begin{proof}
Let
\be
S_{t,\beta}  = \cos(t\tilde{\mc A}_+^{1/2})Q_{0, \beta}
				 + \Pi_0Q_{0,\beta}
			 + \cosh(t\tilde{\mc A}_-^{1/2}) Q_{0,\beta}
\ee
Then $S_{t,\beta}$ is well defined since $\cos(t\tilde{\mc A}_+^{1/2})$ and $\Pi_0$ are bounded operators, and $Q_{0,\beta} \in \dom \cosh(t\tilde{\mc A}_-^{1/2})$ on account of the spectral cutoff \autoref{Qlambda} used in defining $Q_{t,\beta}$. Since $Q_{0, \beta} \in \dom  \tilde{\mc A}$, it follows that $S_{t,\beta} \in \dom \tilde{\mc A}$ for all $t$. Using the facts that $\dom \tilde{\mc A} \subset \dom \tilde{\mc A}_+^{1/2}$ and $\dom \tilde{\mc A} \subset \dom \tilde{\mc A}_-^{1/2}$, it follows that that $S_{t,\beta}$ is twice differentiable in $t$ and
\be
\frac{d^2S_{t,\beta}}{dt^2} = - \tilde{\mc A} S_{t,\beta} \, .
\ee
Furthermore, we have $S_{0,\beta} = Q_{0,\beta}$ and $(d S_{t,\beta}/dt)|_0 = 0$. The proposition now follows immediately by applying \autoref{lem:unique-soln} to $R_t = Q_{t,\beta} - S_{t,\beta}$.

\end{proof}

We now state and prove our main instability result:

\begin{prop} \label{prop:blowup} Let $Q_0 = \mc K P'_0$ with $P'_0 \in \ms V_{\rm odd}^\infty$ be such that $\ms U(Q_0, Q_0) < 0$. Then the solution generated by the initial data $(P_0 = 0, Q_0)$ grows exponentially with time in the sense that there exists $C > 0$ and $\alpha > 0$ such that 
\be
\norm{Q_t}_{\ms H} > C \exp(\alpha t)
\ee
\end{prop} 
\begin{proof}
By \autoref{HU} we have 
\be
\inp{Q_0 , \mc A Q_0}_{\ms H} = \ms U(Q_0 , Q_0) < 0 \, .
\ee
Therefore, there exists $\alpha > 0$ such that $E_{-\alpha^2} Q_0 \neq 0$. Clearly, we can choose $\beta$ sufficiently large and negative so that $E_{-\alpha^2} (I - E_\beta) Q_0 = E_{-\alpha^2} Q_{0, \beta} \neq 0$. However, by \autoref{prop:Qtformula} and the spectral representation \autoref{specrep}, we have
\begin{eqnarray}
\norm{E_{-\alpha^2} Q_{t, \beta}}^2_{\ms H} &=& \int^\beta_{-\alpha^2} \cosh^2 (t \lambda^{1/2}) \, d \norm{E_\lambda Q_0}^2_{\ms H} \nonumber \\
&\geq& \cosh^2 (t \alpha) \norm{E_{-\alpha^2} Q_{0,\beta}}^2_{\ms H} > \exp(2\alpha t) \norm{E_{-\alpha^2} Q_{0,\beta}}^2_{\ms H}
\end{eqnarray}
Thus, we obtain
\be
\norm{Q_t}_{\ms H} \geq \norm{Q_{t, \beta}}_{\ms H} \geq \norm{E_{-\alpha^2} Q_{t, \beta}}_{\ms H} >  \norm{E_{-\alpha^2} Q_{0, \beta}}_{\ms H} \exp(\alpha t)
\ee
as we desired to show.
\end{proof}

\noindent
{\bf Remark}: By similar arguments, it follows that $\norm{\dot{Q}_t}_{\ms H}$ also grows exponentially with $t$. However, we have
\be
\norm{\dot{Q}_t}^2_{\ms H} = \norm{\mc K P_t}^2_{\ms H} = \ms K(P_t, P_t)
\ee
where $P_t$ is the ($t$-$\phi$)-odd part at time $t$ of the solution generated by the initial data $(0, Q_0)$. Thus, the kinetic energy, $\ms K$, of the solution generated by $(0, Q_0)$ blows up exponentially with time. Since the total canonical energy $\ms E$ is conserved, it follows that $-\ms U(Q_t, Q_t)$ also blows up exponentially. Since $\ms K$ and $\ms U$ are gauge invariant, this shows that the exponential blow-up found in \autoref{prop:blowup} is not a gauge artifact.

\medskip

We now reformulate our results as a Rayleigh-Ritz-type of variational principle for analyzing black hole instability, which makes no direct reference to the Hilbert space $\ms H$.

\begin{thm}[Variational Principle for Instability] \label{thm:exp-growth}
For any $P \in {\ms V}_{\rm odd}^\infty$ consider the quantity
\be \label{eq:rate-defn}
\omega^2 (P) \defn \frac{{\ms U}({\mc K} P, {\mc K} P)}{{\ms K} (P,P)} 
\ee
Then if $\omega^2 < 0$, the solution $\delta X(t)$ determined by the initial data $(P,0)$ will grow with time at least as fast as $\exp(\alpha t)$ for any $\alpha < |\omega|$, in the sense that the kinetic energy, $\ms K$, of $\Lie_t \delta X$ will satisfy 
\be
\lim\limits_{t \to \infty} \left[\ms K (\Lie_t \delta X, \Lie_t \delta X) \exp(-2\alpha t)\right] = \infty \, .
\ee
\end{thm}
\begin{proof}
Let $Q_0 = \mc K P$. Then we have
\be
{\ms U}({\mc K} P, {\mc K} P) = \ms U (Q_0, Q_0) = \inp{Q_0, \mc A Q_0}_{\ms H}
\ee
and
\be
\ms K (P,P) = \inp{Q_0, Q_0}_{\ms H} \, .
\ee
Thus, 
\be
	\omega^2 = \frac{ \inp{Q_0, \mc A Q_0}_{\ms H} }{ \inp{Q_0, Q_0}_{\ms H} } 
\ee
so if $\omega^2 < 0$, then it must be the case that $E_{-\alpha^2} Q_0 \neq 0$ for any $\alpha^2 < |\omega^2|$.

The solution $\Lie_t \delta X$ has initial data $(0,Q_0)$. The theorem follows immediately by applying \autoref{prop:blowup} and the remark below that proposition to this solution.

\end{proof}

\noindent
{\bf Remark}: The variational principle in \autoref{thm:exp-growth} may be viewed as a generalization of the variational principle of \cite{SW} for spherically symmetric perturbations of static, spherically symmetric spacetimes (with matter fields), which itself is a generalization of the variational principle of \cite{Chandra}. Indeed, much more generally, in an arbitrary theory, the key ingredients needed to obtain such a variational principle for linearized perturbations off of a stationary background are that (i) the theory be derived from a Lagrangian, so that the canonical energy, $\ms E$, of perturbations is well defined; (ii) gauge conditions can be chosen that uniquely fix the gauge in such a way that $\ms E$ acts as a Hamiltonian for the linearized perturbations and time evolution yields a suitably smooth map on the initial data space; (iii) the background admits a suitable \(t\)-\(\phi\) reflection isometry, $i$, so that the canonical energy $\ms E$ can be decomposed into a kinetic part $\ms K$ and a potential part $\ms U$; (iv) the kinetic energy $\ms K$ is positive definite. Thus, it should be possible to generalize our results to a wide variety of other cases as well as a wide variety of other theories. As simple examples, in Appendix \ref{sec:other-fields} we analyze the cases of an axisymmetric Klein-Gordon scalar field and an axisymmetric electromagnetic field propagating in an arbitrary stationary-axisymmetric black hole background with a ($t$-$\phi$)-reflection isometry.

\bigskip

{\bf Acknowledgements}

We wish to thank Stefan Hollands for numerous extremely helpful discussions and communications throughout the course of this work. This research was supported in part by NSF grant PHY 12-02718 to the University of Chicago.

\appendix


\section{Projection Operators}\label{sec:proj-ops}

The purpose of this Appendix is to outline the proof of the following proposition:

\begin{prop} \label{prop:proj} The orthogonal projection operators $\Pi_c: \ms P_{\rm gr} \to \ms P_{\rm gr}$ and $\Pi_g: \ms P_{\rm gr} \to \ms P_{\rm gr}$ defined in \autoref{sec:lin-pert} map $\ms P_{\rm gr}^\infty$ to itself and---as linear maps on $\ms P_{\rm gr}^\infty$---are continuous in the natural (Fr\'{e}chet) topology on $\ms P_{\rm gr}^\infty$.
\end{prop}

The basic strategy for proving this proposition consists of constructing these projection operators by solving elliptic equations. The nice properties of the projection operators then follow from elliptic regularity. 
This strategy was successfully implemented in Proposition 5 of \cite{HW-stab}, with the main technical steps carried out in Lemma 3 of that reference. However, the choice of boundary conditions for spaces our $\ms W_c$ and $\ms W_g$ differ from the space, $\ms W_{HW}$, considered in \cite{HW-stab}, and there are several other (relatively minor) differences as well as some (small) improvements that can be made. Therefore, rather than merely referring the reader to \cite{HW-stab}, we shall we shall now outline the main steps needed to prove \autoref{prop:proj}. However, our proof mirrors the proof given in \cite{HW-stab} in its essential details.

To project $\delta X \in \ms P_{\rm gr}$ orthogonally to $\mc S [\ms W_c]$, with $\ms W_c$ defined by  \autoref{eq:Wc-defn}, we would like to solve the elliptic equation
\be
\mc L \mc L^* \svp \xi = \mc L\delta X 
\label{Pic}
\ee
with $\svp \xi$ satisfying the conditions appearing in the definition \autoref{eq:Wc-defn} of $\ms W_c$ and also such that $\delta X - \mc L^* \svp \xi$ satisfies the boundary conditions \autoref{admcond} and \autoref{horcond}. If we can find such a $\xi$, then
	\be\label{eq:Pic-defn}
		\delta X' \defn \delta X - \mc L^*\svp\xi
	\ee
will satisfy $\mc L \delta X' = 0$ together with the boundary conditions \autoref{admcond} and \autoref{horcond}. Thus, we have $\delta X' \in \ms V_c$ whereas $\mc S^* \mc L^* \xi \in \ms W_c$, so the desired projection map is 
\be
\Pi_c \delta X = \delta X' = \delta X - \mc L^*\svp\xi\, .
\ee
Similarly, to obtain the projection operator $\Pi_g$, we we wish to solve the equation 
\be
\mc L\mc L^*\svp\xi = \mc L\mc S\delta X
\ee
with $\svp \xi$ satisfying the conditions appearing in the definition \autoref{eq:Wg-defn} of $\ms W_g$ and also such that $\delta X - \mc S^* L^* \svp \xi$ satisfies the boundary conditions \autoref{eq:gauge-conds-adm} and \autoref{eq:gauge-conds-B}. We then obtain
	\be\label{eq:Pig-defn}
	\Pi_g \delta X = \delta X - \mc S^*\mc L^*\svp\xi \, .
	\ee

Thus, the proof hinges on the ability to solve the equation
$\mc L \mc L^* \svp \xi = \svp \alpha$
for suitable $\svp \alpha$ such that $\svp \xi$ satisfies the desired boundary conditions.
The key step in showing this is the proof of a Poincare inequality of the form
	\be\label{eq:bounded-coercive}
		c\norm{\svp\xi}_{W^2_\rho \oplus W^1_\rho} \leq \norm{\mc L^* \mc M_\rho \svp\xi}_{L^2 \oplus L^2} \leq C\norm{\svp\xi}_{W^2_\rho \oplus W^1_\rho}
	\ee
for positive constants \(c, C\) where $W^2_\rho \oplus W^1_\rho$ denotes the weighted Sobolev space on $\svp \xi = (\xi, \xi^a)$ and the weighting matrix appearing in the middle term is
\be
 \mc M_\rho \defn \begin{pmatrix}\rho^2  & 0 \\ 0 & \rho \end{pmatrix} \, ,
 \ee 
 where $\rho$ is a positive function that goes to $1$ in a neighborhood of $B$ and approaches \( \lb( x_1^2 + \ldots + x_d^2 \rb)^{1/2} \) near infinity. We want this inequality to hold for all $\svp \xi$ lying in a suitable closed subspace of $W^2_\rho \oplus W^1_\rho$. As can be seen from the proof of lemma 3 of \cite{HW-stab}, a certain boundary term at $B$ arises in the estimates, and, for the validity of the first inequality in \autoref{eq:bounded-coercive}, we need $\svp \xi$ to be such that this boundary term vanishes. This requires the vanishing of
	\be\begin{split}
		  & \int_Br^a \lb( D_b\xi_a \xi^b - D_b\xi^b \xi_a \rb) = \int_B \lb( \xi^aD_a\xi_r  - \xi_rD_a\xi^a  \rb) \\
		= & \int_B\lb( \xi_rd_r\xi_r  + \xi^A\ms D_A \xi_r  - \xi_rd_r\xi_r  + \xi_r\ms D_A\xi^A  \rb) \\
		= & \int_B\lb( \xi^A\ms D_A \xi_r  -  \xi_r\ms D_A\xi^A  \rb) = - 2 \int_B \xi_r \ms D_A\xi^A 
	\end{split}\ee
Thus, the Poincare inequality \autoref{eq:bounded-coercive} will hold if we restrict $\svp \xi$ to the closed subspace of $W^2_\rho \oplus W^1_\rho$ defined by\footnote{This condition is well defined since for any $\xi^a \in W^1_\rho$, its restriction to $B$ is well defined in $L^2(B)$ and its tangential derivatives are therefore also well defined (weakly).} $(\ms D_A\xi^A)|_B = 0$ (or alternatively, the closed subspace defined by $\xi_r = 0$). In addition, we may impose the vanishing at $B$ of any other components of $\svp \xi$ and $\xi_r$ and/or their tangential derivatives. However, we cannot impose conditions on radial derivatives of these quantities, since such conditions would not define closed subspaces of $W^2_\rho \oplus W^1_\rho$.

The closed subspaces of $W^2_\rho \oplus W^1_\rho$ of interest are
\begin{subequations}\label{eq:X-defn}\begin{align}
\ms X_{HW} &\defn \lb\{ \svp\xi = (\xi,\xi^a) \in W^2_\rho \oplus W^1_\rho \st \xi\vert_B  = \xi_r\vert_B \text{ and } \xi^A\vert_B = 0 \rb\} \\
\ms X_{c}  &\defn \lb\{ \svp\xi \in W^2_\rho \oplus W^1_\rho \st \xi^A\vert_B = 0 \rb\} \\
\ms X_{g}  &\defn \lb\{ \svp\xi \in W^2_\rho \oplus W^1_\rho \st \xi\vert_B = \xi_r\vert_B = (\ms D_A\xi^A)|_B = 0 \rb\}
\end{align}\end{subequations}
The space $\ms X_{c}$ will be used for the construction of $\Pi_c$ and the space $\ms X_{g}$ will be used for the construction of $\Pi_g$. The space $\ms X_{HW}$ is included in this list because it is the space used in \cite{HW-stab}, but it will not be used here.

As shown in the proof of lemma 3 of \cite{HW-stab}, for $d \ge 4$ (i.e., spacetime dimension $5$ or greater), the Poincare inequality \autoref{eq:bounded-coercive} holds on the spaces $\ms X_{HW}$, $\ms X_{c}$, and $\ms X_{g}$. Equation \autoref{eq:bounded-coercive} also holds in $d=3$ provided that, if the black hole is non-rotating (i.e., for the Schwarzschild case), we pass to the subspace orthogonal in $L^2 \oplus L^2$ to $\mc M_\rho^{-1} t^\mu$, where $t^\mu$ is the timelike Killing field of the background black hole. 

Now, let $\ms X$ denote any of the spaces $\ms X_{HW}$, $\ms X_{c}$, and $\ms X_{g}$, with the additional projection orthogonal to $\mc M_\rho^{-1} t^\mu$ taken in the case of a Schwarzschild background in $d=3$. In view of \autoref{eq:bounded-coercive}, the Lax-Milgram theorem\footnote{This theorem is just the Riesz lemma applied to $\inp{\svp \beta, \cdot}$,
which is a bounded linear map on $W^2_\rho \oplus W^1_\rho$ and hence, by \autoref{eq:bounded-coercive}, is also a bounded linear map on the inner product space defined by $\inp{\mc L^* \mc M_\rho \, \cdot, \mc L^* \mc M_\rho \, \cdot} $.} 
implies that given any \(\svp \beta \in \ms X\), there exists $\svp\eta \in \ms X$ such that for all $\svp\chi \in \ms X$ we have
	\be\label{eq:LL*-weak}
		\inp{\mc L^* \mc M_\rho \svp\eta, \mc L^* \mc M_\rho \svp\chi } = \inp{ \svp \beta, \svp\chi }  \, ,
	\ee
where $\inp{,}$ denotes the $L^2$ inner product. In particular, since this equation holds for all $\svp \chi \in C_0^\infty$, it follows that 
\be
\xi \defn  \mc M_\rho \svp \eta 
\ee
is a weak solution to
\be
\mc L \mc L^*\svp\xi = \svp \alpha
\label{LLeq}
\ee
where $\alpha =  \mc M_\rho \svp \beta$.
However, since $\mc L \mc L^*$ is elliptic (see \cite{Chr-De}, \cite{Corvino} and also Ch.VII Sec.2 of \cite{CB-book}), if $\svp \beta \in \ms X \inter_k (W^k_\rho \oplus W^k_\rho)$ (so that $\svp \alpha \in (\mc M_\rho \ms X) \inter_k  (\rho^2 W^k_\rho \oplus \rho W^k_\rho)$), then $\svp \xi \in (\mc M_\rho \ms X) \inter_k  (\rho^2 W^k_\rho \oplus \rho W^k_\rho)$ and \autoref{LLeq} holds in the ordinary (strong) sense. Furthermore $\svp \xi$ depends continuously on $\svp \alpha$ in the natural (Fr\'{e}chet) topology of these spaces. It also follows from \autoref{eq:LL*-weak} together with the fact that \autoref{LLeq} holds in the strong sense that for all $\svp \alpha, \svp \psi \in (\mc M_\rho \ms X)  \inter_k (\rho^2 W^k_\rho \oplus \rho W^k_\rho)$, the solution $\svp \xi$ to \autoref{LLeq} given by the Lax-Milgram theorem satisfies
\be
\inp{\mc L^* \mc \svp \xi, \mc L^* \svp\psi } = \inp{ \mc L \mc L^*\svp\xi, \svp\psi }
\ee
However, since $\svp \xi$ and $\svp \psi$ are smooth, the difference between the left and right sides can be explicitly computed by the same calculation as previously done in \autoref{eq:bound-integral}. In this case, no boundary terms arise from infinity and we obtain
	\be
\inp{\mc L^* \mc  \svp\xi, \mc L^* \svp\psi } - \inp{ \mc L \mc L^*\svp\xi, \svp\psi } = 2\int_B \lb( -d_r\psi~ \delta\varepsilon + \psi~\delta\vartheta_{\rm even} + \psi_r~ \delta\vartheta_{\rm odd} - \psi^{AB}~\delta\varpi_{AB} \rb)
	\ee
where $\delta\varpi_{AB}$ is given by \autoref{eq:varpi-defn} and we have used \(\ms D_A\psi^A = 0\) to replace $\psi^A$ by $\ms D_B\psi^{AB}$ for some \(\psi^{AB} = \psi^{[AB]}\). Since the right side vanishes for all $\svp \psi \in (\mc M_\rho \ms X)  \inter_k (\rho^2 W^k_\rho \oplus \rho W^k_\rho)$, we obtain the following additional boundary conditions on the solutions $\svp \xi$ in the various cases:
\begin{itemize}
		\item \(\ms X_{HW} \): {\em Definition}: \(\xi\vert_B = \xi_r\vert_B\), \( \xi^A\vert_B = 0 \); {\em Additional conditions on solutions}: $\delta\varepsilon \vert_B = 0$,  $\delta\vartheta_+\vert_B = 0$.
		\item \(\ms X_c \): {\em Definition}: \( \xi^A\vert_B = 0 \); {\em Additional conditions on solutions}: $\delta\varepsilon \vert_B = 0$, $\delta\vartheta_{\rm even} \vert_B = 0$, $\delta\vartheta_{\rm odd}\vert_B = 0$.
		\item \(\ms X_g \): {\em Definition}: \(\xi\vert_B = 0 = \xi_r\vert_B\), \( \ms D_A\xi^A\vert_B = 0 \); {\em Additional conditions on solutions}: $\delta\varepsilon \vert_B = 0$, $\delta\varpi_{AB} \vert_B = 0$.
\end{itemize}
Thus, the solutions given by the Lax-Milgram theorem to \autoref{LLeq} with $\svp \alpha \in (\mc M_\rho \ms X) \inter_k  (\rho^2 W^k_\rho \oplus \rho W^k_\rho)$ satisfy precisely the conditions at $B$ needed to obtain the spaces $\ms W_{HW}$, $\ms W_{c}$, and $\ms W_{g}$, respectively.

The key ingredients are now in place to construct the projection map $\Pi_c$. As explained at the beginning of this section, to construct $\Pi_c$, we wish to solve \autoref{Pic} for any given $\delta X \in \ms P^\infty_{\rm gr}$. In order to obtain a solution that satisfies the desired boundary conditions at $B$, we choose $\svp {\hat{\xi}}$ be smooth and of compact support such that: (i) $\hat{\xi}^A\vert_B = 0$; (ii) $\mc L^* \svp {\hat{\xi}} - \delta X$ satisfies the horizon boundary conditions \autoref{horcond} required for elements of $\ms V_c$; and (iii) $\lb(\mc L \mc L^* \svp {\hat{\xi}} - \mc L \delta X \rb)^A \vert_B = 0$. There is no difficulty simultaneously satisfying all of these conditions because $\svp {\hat{\xi}}$ is not constrained by any equations holding in the bulk, so all of the various radial derivatives $(d_r)^n \svp {\hat{\xi}}$ at $B$ can be chosen independently.

We now solve
\be
\mc L \mc L^* \svp {\tilde{\xi}}  = \mc L\delta X - \mc L \mc L^* \svp {\hat{\xi}} 
\ee
Condition (iii) guarantees that the right side lies in $(\mc M_\rho \ms X_c)  \inter_k (\rho^2 W^k_\rho \oplus \rho W^k_\rho)$, so we obtain a solution $\svp {\tilde{\xi}} \in (\mc M_\rho \ms X_c)  \inter_k (\rho^2 W^k_\rho \oplus \rho W^k_\rho)$. 
It follows that $\delta X - \mc L^* (\svp {\tilde{\xi}} + \svp {\hat{\xi}})$ satisfies the constraints \autoref{linconst} and the horizon boundary conditions \autoref{horcond} needed to be an element of $\ms V_c^\infty$. However, it need not satisfy the remaining conditions \autoref{admcond}, namely $\delta M = \delta J_\Lambda = \delta P_i = 0$. Since these are only a finite number of conditions, only a finite co-rank further projection is needed. 

In \cite{HW-stab}, this additional projection was done by choosing an arbitrary minimal set of solutions $\Psi_I$ that satisfy the horizon boundary conditions and whose span yields arbitrary values of $(\delta M, \delta J_\Lambda, \delta P_i)$. The final projection to $\ms V_c^\infty$ was done by subtracting from $\delta X - \mc L^* (\svp {\tilde{\xi}} + \svp {\hat{\xi}})$ the linear combination of $\Psi_I$ that makes $\delta M = \delta J_\Lambda = \delta P_i = 0$. However, this does not, in general, yield an orthogonal projection, as would be needed to obtain commutativity with $\Pi_g$. 

We will carry out this additional projection by making a special choice of $\Psi_I$ that yields an orthogonal projection. Let $\svp {\hat {\chi}}_I$ be a basis for asymptotic translations and rotations with respect to the axial Killing fields at infinity, with each $\svp {\hat {\chi}}_I$ chosen so that it vanishes in a neighborhood of $B$. Now solve the equation
\be
\mc L \mc L^* \svp {\tilde{\chi}}_I = - \mc L \mc L^* \svp {\hat{\chi}}_I 
\ee
to obtain a solution in $(\mc M_\rho \ms X_c)  \inter_k (\rho^2 W^k_\rho \oplus \rho W^k_\rho)$. Define 
\be
\svp \xi_I = \svp {\tilde{\chi}}_I + \svp {\hat{\chi}}_I \, .
\ee
It follows immediately from the definition \autoref{eq:Wc-defn} of $\ms W_c$ that $S^* \mc L^* \svp \xi_I \in \ms W_c$. Any element of $\ms V_c$ must therefore be $L^2$-orthogonal to $\mc L^* \svp \xi_I$. In particular, we have $\mc L^* \svp \xi_I \notin \ms V_c$. But $\mc L^* \svp \xi_I$ satisfies the constraints \autoref{linconst} and the horizon boundary conditions \autoref{horcond} needed to be an element of $\ms V_c^\infty$. This implies that $\mc L^* \svp \xi_I$ cannot satisfy $\delta M = \delta J_\Lambda = \delta P_i = 0$. More generally, no linear combination of the $\mc L^* \svp \xi_I$ can satisfy $\delta M = \delta J_\Lambda = \delta P_i = 0$, so the span of $\mc L^* \svp \xi_I$ must yield all possible values of $(\delta M, \delta J_\Lambda, \delta P_i)$. Thus, $\mc L^* \svp \xi_I$ may be used to play the role of $\Psi_I$.

Now let 
\be
\svp \xi = \svp {\tilde{\xi}} + \svp {\hat{\xi}} + \sum_I c_I \svp \xi_I
\ee
where the $c_I$ are chosen so that $\delta X - \mc L^* \svp \xi$ satisfies $\delta M = \delta J_\Lambda = \delta P_i = 0$. Putting everything together, we have shown that any $\delta X \in \ms P_{\rm gr}^\infty$ can be written as
\be
\delta X = \delta X' + \mc L^* \svp \xi
\ee
where $S^* \mc L^* \svp \xi \in \ms W_c$ and $\delta X' \in \ms V_c^\infty$. Thus, we have $\Pi_c \delta X = \delta X' \in \ms V_c^\infty$, as we desired to show. The continuous dependence of $\delta X'$ on $\delta X$ in the topology of $\ms P_{\rm gr}^\infty$ follows directly from the continuous dependence of the solution $\svp \xi$ on $\svp \alpha$ in \autoref{LLeq}.

The corresponding results for the projection map $\Pi_g$ are obtained in complete parallel by replacing $\ms W_c$ and $\ms X_c$ with $\ms W_g$ and $\ms X_g$ and inserting the map $\mc S$ in appropriate places.

\section{Scalar and Electromagnetic Fields}\label{sec:other-fields}

In this Appendix, we give a treatment of scalar fields and electromagnetic fields on a black hole background analogous to our treatment of gravitational perturbations. The black hole background is assumed to be stationary and axisymmetric with a ($t$-$\phi$) reflection isometry $i$ and with a bifurcate Killing horizon with bifurcation surface $B$, but it need not be a solution to Einstein's equation. As in the gravitational case, we consider only axisymmetric scalar and electromagnetic fields.

\subsection{Scalar Fields}

The case of a Klein-Gordon scalar field $\phi$ propagating on a black hole background is particularly simple because there are no constraints and there is no gauge freedom. The initial data for a scalar field is \(\delta X = \lb(\pi,\psi\rb)\), where $\psi = \phi|_\Sigma$ and $\pi = (u^\mu \nabla_\mu \phi)|_\Sigma$, where $u^\mu$ is the unit normal to $\Sigma$. The analog of $\ms P_{\rm gr}$ is
\be
\ms P_{\rm KG} \defn L^2(\Sigma) \oplus L^2(\Sigma) 
\ee
with the natural volume measure on $\Sigma$. The analog, $\ms P_{\rm KG}^\infty$, of $\ms P_{\rm gr}^\infty$ is the intersection of the corresponding weighted Sobolev spaces (see \autoref{Pgrinfty}). The symplectic form is
\be
\Omega_\Sigma (\widetilde{\delta X},\delta X) \defn \int_\Sigma (\tilde \pi \psi - \tilde \psi \pi) \, ,
\ee
which is represented on $\ms P_{\rm KG}$ by the orthogonal linear map \(\mc{S}\) given by
\be
\mc S (\pi , \psi) = ( \psi, -\pi) \, .
\ee
Since we have no constraints or gauge conditions to enforce, the analogs of the projectors $\Pi_c$ and $\Pi_g$ are trivial, so $\ms V^\infty_{\rm KG} = \ms P_{\rm KG}^\infty$.

In both the static and rotating cases, the reflection odd and even parts, respectively, of the perturbation are simply \(P = \lb( \pi,  0 \rb)\) and \(Q = \lb( 0,  \psi\rb)\). Note that in this case we have $\ms V^\infty_{\rm KG,odd} \cong \ms V^\infty_{\rm KG,even}$. The kinetic and potential energies for axisymmetric fields are:
	\be\begin{split}
		\ms K_{\rm KG} &= \int_\Sigma N\pi^2 \\
		\ms U_{\rm KG} &= \int_\Sigma N\lb[ (D_a\psi)^2 + m^2\psi^2\rb]
	\end{split}\ee
Both $\ms K_{\rm KG}$ and $\ms U_{\rm KG}$ are manifestly positive definite. Indeed, for a scalar field\footnote{For other fields the canonical energy can differ from the right side of \autoref{canstress} by a ``boundary term''; see the appendix of \cite{IW} for further discussion.}, the canonical energy takes the form
\be
\ms E_{\rm KG} = \ms K_{\rm KG} + \ms U_{\rm KG} = 2\int_\Sigma T_{\mu\nu}t^\mu u^\nu 
\label{canstress}
\ee
where \(T_{\mu\nu}\) is the energy-momentum tensor
\be
	T_{\mu\nu} \defn \nabla_\mu\phi \nabla_\nu\phi -\frac{1}{2} g_{\mu\nu}\lb( \nabla_\lambda\phi\nabla^\lambda\phi + m^2\phi^2 \rb)
\ee
For axisymmetric fields we then have \(T_{\mu\nu}t^\mu u^\nu = N T_{\mu\nu}u^\mu u^\nu \geq 0 \), which again shows the positivity of the canonical energy.

The evolution equations for axisymmetric fields are
	\be\label{eq:evol-KG}\begin{split}
		\dot\psi & = N\pi  \\
		\dot\pi  & = D_a\lb( ND^a\psi \rb) - Nm^2\psi
	\end{split}\ee
Thus, the operator $\mc K_{\rm KG}$ corresponds simply to multiplication by $N$. It is then obvious that for $\psi \in \ms V^\infty_{\rm KG,even}$, we have \(\psi \in \mc K_{\rm KG}[\ms V_{\rm KG,odd}^\infty]\) if and only if \(\psi\vert_B = 0\). The ``inverse kinetic energy Hilbert space'' $\ms H_{\rm KG}$ is then the $L^2$-space on $\Sigma$ with volume measure given by $N^{-1}$ times the natural volume measure on $\Sigma$.

Since $\ms U_{\rm KG}$ is positive definite, the operator $\mc A_{\rm KG} : \ms H_{\rm KG} \to \ms H_{\rm KG}$ is positive definite, so there cannot exist any exponential growth instabilities for axisymmetric\footnote{For a static black hole background, the restriction to axisymmetric perturbations is unnecessary.} scalar fields in any black hole background. In fact, in this case the expression for $\ms U_{\rm KG}$ is sufficiently simple that one can prove stability for elements of $\ms V^\infty_{\rm KG,even}$ that lie in $\ms H_{\rm KG}$, as done in \cite{W-stab} for the case of Schwarzschild. The restriction to data vanishing at $B$ can then be eliminated (for general black holes) by the ``trick'' used in \cite{KW-stab}. For the case of a Schwarzschild or Kerr black hole, these results can be greatly improved by the methods of \cite{Daf-Rod-lec} and \cite{DRS-stab}.

\subsection{Electromagnetic Fields}

The initial data for a Maxwell field is \(\delta X = \lb(E_a,  A_a\rb)\), where $E_a$ is the electric field on $\Sigma$ and $A_a$ is the pullback of the vector potential to $\Sigma$. The analog of $\ms P_{\rm gr}$ is
\be
\ms P_{\rm EM} \defn L^2(\Sigma, T^*\Sigma)\oplus L^2(\Sigma, T^*\Sigma) 
\ee
and the analog, $\ms P_{\rm EM}^\infty$, of $\ms P_{\rm gr}^\infty$ is the intersection of the corresponding weighted Sobolev spaces. The symplectic form is
\be
\Omega_\Sigma (\widetilde{\delta X},\delta X) \defn \int_\Sigma (\tilde E^a A_a - \tilde A_a E^a) \, ,
\ee
which is represented on $\ms P_{\rm EM}$ by the orthogonal linear map \(\mc{S}\) given by
\be
\mc S (E^a , A_a) = ( A^a, -E_a) \, .
\ee

The analog of the operator $\mc L$ of \autoref{eq:L-defn} maps smooth initial data on $\Sigma$ to smooth functions on $\Sigma$ and is given by
\be\label{eq:L-defn-EM}
\mc L_{\rm EM} \begin{pmatrix}E_a\\ A_a\end{pmatrix} \defn D^aE_a \, ,
\ee
i.e., the constraints are
\be
\mf c_{\rm EM} = \mc L_{\rm EM} \delta X = D^aE_a = 0 \, .
\label{emconst}
\ee
The formal $L^2$ adjoint of $\mc L_{\rm EM}$ maps smooth functions on $\Sigma$ to smooth initial data on $\Sigma$ and is given by
\be\label{eq:L*-defn-EM}\begin{split}
\mc L_{\rm EM}^*\xi = \begin{pmatrix} -D_a\xi \\ 0\end{pmatrix} \, .
\end{split}\ee
The gauge transformations are then \(\hat\delta_\xi X = \mc S^*\mc L_{\rm EM}^*\xi \).

The analog of the space $\ms W_c$ of \autoref{eq:Wc-defn} is
	\be
		(\ms W_{\rm EM})_c \defn \{ S^*\mc L_{\rm EM}^* \xi \in \ms P_{\rm EM}^\infty \st \xi\vert_B = 0 \text{ and } \xi \sim {\rm const.}\vert_\infty \}
	\ee
The symplectically-orthogonal space is \((\ms V_{\rm EM})_c \defn (\ms W_{\rm EM})_c^{\mc S\perp}\). If $\delta X \in \ms P_{\rm EM}^\infty$, then $\delta X \in (\ms V_{\rm EM})_c^\infty \defn (\ms V_{\rm EM})_c \cap \ms P_{\rm EM}^\infty$ if and only if
	\be\label{eq:const-conds-EM}\begin{split}
	\mf c_{\rm EM}  &= 0 \\
	\mc Q &= 0
\end{split}\ee
where
\be
\mc Q \defn \int_\infty E_\rho
\ee
is the electric charge of the solution. In view of \autoref{emconst}, the charge integral may be taken over any sphere homologous to a sphere at infinity.

The analog of the space $\ms W_g$ of \autoref{eq:Wg-defn} is
	\be
		(\ms W_{\rm EM})_g \defn \{ S^*\mc L_{\rm EM}^* \xi \in \ms P_{\rm EM}^\infty \}
	\ee
i.e., $(\ms W_{\rm EM})_g$ consists of all gauge transformations lying in $\ms P_{\rm EM}^\infty$. The space $(\ms V_{\rm EM})_g$ is the orthogonal complement of $(\ms W_{\rm EM})_g$. Elements of $(\ms V_{\rm EM})_g^\infty \defn (\ms V_{\rm EM})_g \cap \ms P_{\rm EM}^\infty$ satisfy the gauge conditions
\be\label{eq:gauge-conds-EM}\begin{split}
\mf g_{\rm EM} \defn \mc L_{\rm EM} \mc S \delta X = D^aA_a = 0 \\
 A_r\vert_B = 0 \, .
\end{split}\ee
These conditions fix the gauge completely. The projection operators $(\Pi_{\rm EM})_c$ and $(\Pi_{\rm EM})_g$ can now be defined in direct analogy with $\Pi_c$ and $\Pi_g$ and they satisfy the analogous properties. Indeed, in the Maxwell case, the operator
\(\mc L_{\rm EM} \mc L_{\rm EM}^* = - \triangle\) is just the Laplacian operator acting on scalar functions, so the analogs of many of the results of Appendix \ref{sec:proj-ops} are standard.

The perturbations in
\be
 \ms V_{\rm EM}^\infty \defn  (\Pi_{\rm EM})_g (\Pi_{\rm EM})_c \ms P_{\rm EM}^\infty 
\ee 
can again be split into their odd and even parts under the reflection symmetry $i$. In the case of static black hole background, we may take $i$ to be the \(t\)-reflection isometry. The reflection odd and even parts of the initial data are then \(P = \lb(E_a,  0\rb)\) and \(Q  = \lb( 0,  A_a\rb)\), respectively. The kinetic and potential energies are given by
	\be\begin{split}
		\ms K_{\rm EM} & = \int_\Sigma N E_a^2  \\
		\ms U_{\rm EM} & = 2\int_\Sigma N (D_{[a}A_{b]})^2
	\end{split}\ee
Both $\ms K_{\rm EM}$ and $\ms U_{\rm EM}$ are manifestly positive-definite for any static black hole background. Thus, as for scalar fields, there cannot exist any exponential growth instabilities for electromagnetic fields propagating on an arbitrary static black hole background.

However, the situation is quite different for axisymmetric electromagnetic fields propagating on a stationary-axisymmetric black hole background. We can decompose $E_a$ and $A_a$ into their $\phi$-reflection odd and even parts as
\be
E_a = E^\Lambda \phi_{\Lambda a} + \tilde E_a
\ee
\be
A_a = A^\Lambda \phi_{\Lambda a} + \tilde A_a
\ee
with $ \tilde E^a  \phi_{\Lambda a} = \tilde A^a  \phi_{\Lambda a} = 0$. 
The reflection-odd and even parts of the initial data are then, respectively,
\be
P_{\rm EM} = \lb(\tilde E_a, -A^\Lambda \phi_{\Lambda a} \rb) 
\ee
\be
Q_{\rm EM} = \lb(E^\Lambda \phi_{\Lambda a}, \tilde A_a \rb) 
\ee
The kinetic and potential energies are then given by
	\be\begin{split} \label{EMenergy}
		\ms K_{\rm EM} & = \int_\Sigma \lb[ N\tilde E_a^2 + 2N \lb(D_{[a}(A^\Lambda\phi_{b]\Lambda})\rb)^2  +4\bar N^\Theta \phi^a_\Theta\tilde E^bD_{[a}(A^\Lambda\phi_{b]\Lambda})  \rb] \\
		\ms U_{\rm EM} & = \int_\Sigma N\lb[ E^\Lambda E_\Lambda + 2 (D_{[a}\tilde A_{b]})^2  \rb] \, .
	\end{split}\ee
It is clear that there exist black hole backgrounds for which $\ms K_{\rm EM}$ fails to be positive\footnote{This possibility was first pointed out to us by A. Ishibashi.}. To see this, we simply start with any black hole background spacetime (e.g., Schwarzschild) and any initial data $(E_a,  A_a) \in \ms V_{\rm EM}^\infty$ such that $\phi^a_\Theta\tilde E^bD_{[a}(A^\Lambda\phi_{b]\Lambda}) \neq 0$ at some point $x \in \Sigma$. Then we can choose a function $\bar N'^\Theta$ that is sufficiently large in a neighborhood of $x$ so that if we replace $\bar N^\Theta$ by $\bar N'^\Theta$ then the last term in the expression for $\ms K_{\rm EM}$  will be negative and will dominate the first two terms. We can then construct a new stationary-axisymmetric black hole spacetime that has the same induced spatial metric on $\Sigma$, but with a new stationary Killing field $t'^\mu$ that has the same lapse function $N$ as $t^\mu$ but has shift vector $\bar N'^\Lambda \phi_{\Lambda a}$. On this new black hole spacetime, $\ms K_{\rm EM}$ will be negative for the original initial data $(E_a,  A_a)$. Of course, the new stationary black hole spacetime will not be a solution to Einstein's equation. It is not obvious whether $\ms K_{\rm EM}$ can be made negative for any black hole background that is a solution to Einstein's equation. 

The evolution equations for the Maxwell field can be written as:
\be\label{eq:evol-EM}\begin{split}
		\dot A_a & = - D_a\varphi + NE_a - 2N^bD_{[a}A_{b]} \\
		\dot E_a & = 2D^b\lb( ND_{[b}A_{a]} + N_{[b}E_{a]}\rb)
\end{split}\ee
where the gauge conditions \autoref{eq:gauge-conds-EM} applied to \(\dot{\delta X}\) imply that $\varphi$ satisfies:
\be\begin{split}
	- \triangle\varphi + D^a(NE_a - 2N^bD_{[a}A_{b]}) & = 0 \\
	(-d_r\varphi - 2r^aN^bD_{[a}A_{b]})\vert_B & = 0
\end{split}\ee
The operators \(\mc K_{\rm EM}\) and \(\mc U_{\rm EM}\) can be read off from these equations by taking their even and odd parts under the reflection symmetry $i$.

Since $\ms K_{\rm EM}$ need not be positive definite, we cannot directly apply the results of \autoref{sec:evolution} and \autoref{sec:exp} to the electromagnetic case. Remarkably, however, the potential energy $\ms U_{\rm EM}$ {\em is} positive definite (see \autoref{EMenergy}). Therefore, if we simply make the canonical transformation $(P_{\rm EM}, Q_{\rm EM}) \mapsto (P'_{\rm EM} = Q_{\rm EM}, Q'_{\rm EM} = -P_{\rm EM})$, then all of the results of \autoref{sec:evolution} and \autoref{sec:exp} apply to $(P'_{\rm EM}, Q'_{\rm EM})$. In particular, if $\delta X \in \ms V_{\rm EM}^\infty$ and $\Lie_t \delta X$ has negative kinetic energy, then $\delta X$ grows exponentially with time.

\newpage
\bibliographystyle{unsrt}
\bibliography{../stability}      
\end{document}